\documentclass[10pt,a4paper]{article}

\author{%
  Holger Drees\\
  \small University of Hamburg\\
  \small Department of Mathematics\\
  \small Bundesstra{\ss}e 55, 20146 Hamburg, Germany\\
  \small \texttt{holger.drees@math.uni-hamburg.de}
  \and
  Johan Segers \qquad
  Micha\l~Warcho\l\\
  \small Universit\'e catholique de Louvain\\
  \small Institut de Statistique, Biostatistique et Sciences Actuarielles\\
  \small Voie du Roman Pays 20, B-1348 Louvain-la-Neuve, Belgium\\
  \small \texttt{johan.segers@uclouvain.be}, \texttt{michal.warchol@uclouvain.be}
}

\title{Statistics for Tail Processes of Markov Chains}

\date{\today}

\usepackage[utf8]{inputenc}
\usepackage[authoryear]{natbib}
\usepackage{amsmath}
\usepackage{amsfonts}
\usepackage{amsthm}
\usepackage{amssymb}
\usepackage{bm}
\usepackage{hyperref}
\usepackage{paralist}
\usepackage{graphicx}
\usepackage{booktabs}
\usepackage{enumerate}
\usepackage{a4wide}
\usepackage{color}
\usepackage{url}

\newcommand{\FfA}[1]{\ensuremath{\hat{\bar{F}}^{(\text{f},A_1)}_{n}\left({#1}\right)}}
\newcommand{\FbA}[1]{\ensuremath{\hat{\bar{F}}^{(\text{b},A_1)}_{n}\left({#1}\right)}}

\newcommand{\CDFfA}[1]{\ensuremath{\hat{F}^{(\text{f},A_1)}_{n}\left({#1}\right)}}
\newcommand{\CDFbA}[1]{\ensuremath{\hat{F}^{(\text{b},A_1)}_{n}\left({#1}\right)}}
\newcommand{\CDFbAe}[1]{\ensuremath{\hat{F}^{(\hat b,A_1)}_{n}\left({#1}\right)}}
\newcommand{\CDFmA}[1]{\ensuremath{\hat{F}^{(\text{m},A_1)}_{n}\left({#1}\right)}}
\newcommand{\CDFfB}[1]{\ensuremath{\hat{F}^{(\text{f},B_1)}_{n}\left({#1}\right)}}
\newcommand{\CDFbB}[1]{\ensuremath{\hat{F}^{(\text{b},B_1)}_{n}\left({#1}\right)}}

\newcommand{\Bias}{\operatorname{Bias}}
\newcommand{\SD}{\operatorname{SD}}
\newcommand{\RMSE}{\operatorname{RMSE}}

\newcommand{\abs}[1]{\lvert{#1}\rvert}

\newcommand{\dto}{\rightsquigarrow}

\newcommand{\expec}[1]{\operatorname{E}\left[{#1}\right]}
\newcommand{\texpec}[1]{\operatorname{E}[{#1}]}
\newcommand{\prob}[1]{\operatorname{P}\left[{#1}\right]}
\newcommand{\var}{\operatorname{var}}
\newcommand{\cov}[1]{\operatorname{cov}\left({#1}\right)}
\newcommand{\diff}{\mathrm{d}}
\newcommand{\1}{\bm{1}}
\newcommand{\differ}{\mathrm{d}}

\newcommand{\sign}{\operatorname{sign}}
\newcommand{\ZZ}{\mathbb{Z}}
\newcommand{\RR}{\mathbb{R}}
\newcommand{\NN}{\mathbb{N}}
\newcommand{\reals}{\RR}
\newcommand{\law}[1]{\mathcal{L}\left({#1}\right)}

\newtheorem{theorem}{Theorem}[section]
\newtheorem{lemma}[theorem]{Lemma}

\newtheorem{proposition}[theorem]{Proposition}

\theoremstyle{remark}
\newtheorem{example}[theorem]{Example}
\newtheorem{remark}[theorem]{Remark}

\numberwithin{equation}{section}
\numberwithin{theorem}{section}

\newcommand{\michal}[1]{\begingroup\color{blue}#1\endgroup}

\allowdisplaybreaks

\begin{document}

\maketitle

\begin{abstract}
At high levels, the asymptotic distribution of a stationary, regularly varying Markov chain is conveniently given by its tail process. The latter takes the form of a geometric random walk, the increment distribution depending on the sign of the process at the current state and on the flow of time, either forward or backward. Estimation of the tail process provides a nonparametric approach to analyze extreme values. A duality between the distributions of the forward and backward increments provides additional information that can be exploited in the construction of more efficient estimators. The large-sample distribution of such estimators is derived via empirical process theory for cluster functionals. Their finite-sample performance is evaluated via Monte Carlo simulations involving copula-based Markov models and solutions to stochastic recurrence equations. The estimators are applied to stock price data to study the absence or presence of symmetries in the succession of large gains and losses.
\end{abstract}

\section{Introduction}
\label{introduction}

If serial dependence at high levels is sufficiently strong, extreme values of a stationary time series may arrive in clusters rather than in isolation. This is the case, for instance, for linear time series with heavy-tailed innovations and for solutions of stochastic recurrence equations. If a particular time series model is to be used for prediction at such high levels, it is important to model these clusters well. Think of tail-related risk measures in finance or of return levels in hydrology: a rapid succession of particularly rainy days may be especially dangerous if the capacity of the system to absorp the water is limited.

To judge the quality of fit of a time series model at extreme levels, it is useful to have a benchmark relying on as little model assumptions as possible. A purely nonparametric approach, however, has the drawback that there may be too few data that are sufficiently large. For the purpose of extrapolation, the empirical measure is inadequate.

A solution is to rely on asymptotic theory describing possible limit distributions for the extremes of a time series. If this family of distributions is not too large, one may hope to be able to fit it to actual data.

For extremes of stationary time series, there are several asymptotic frameworks available, all of them more or less equivalent. For the study of short-range extremal dependence, the tail process \citep{basrak:segers:2009} is a convenient choice. It captures the collection of finite-dimensional limit distributions of the series conditionally on the event that, at a particular time instant, the series is far from the origin. For instance, the tail process determines the tail dependence coefficients and the extremal index \citep{Leadbetter1983}. It is also related to other tail-related objects such as the extremogram \citep{davis:mikosch:2009} and the extremal dependence measure \citep{Resnic2012}.

The family of tail processes of regularly varying time series is still too large to permit accurate nonparametric estimation. Additional assumptions serve to render the inference problem more manageable. The choice made in this paper is to focus on stationary univariate Markov chains. The joint distribution of such a chain is determined by its bivariate margins, yielding considerable simplifications. Its tail process takes the form of a geometric random walk, the increments depending both on the sign of the process at the current state and on the direction of time, forward or backward. The random walk representation goes back to \cite{smith:1992} and was developed further in \cite{Perfekt1997}, \cite{bortot2000}, and \cite{yun2000}. The formulation in terms of the tail process stems from \cite{segersmarkov} and \cite{segersjanssen}. By a marginal standardization procedure, the tail process may also be used for time series with light-tailed margins. Such time series arise for instance in environmental applications.

The tail process of a stationary time series is itself not stationary because of the special role played by the time instant figuring in the conditioning event. Still, its finite-dimensional distributions satisfy a collection of identities regarding the effect of a time shift. These equations can be summarized into the so-called time-change formula; see equation~\eqref{eq:timechange} below. Apart from being a probabilistic nicety, the time-change formula is useful from a statistical perspective because it provides additional information on the distribution of the tail process. Exploiting this information can lead to more efficient inference.

Our contribution is to propose and study nonparametric estimators for the tail process of a stationary univariate Markov chain. Large-sample theory and Monte Carlo simulations both confirm that efficiency gains are possible when the time-change formula is incorporated into the estimation procedure. The asymptotic distributions of the estimators are described via functional central limit theorems building on the empirical process theory developed in \cite{drees2010limit}. The finite-sample performance is investigated for solutions of stochastic recurrence equations and for copula-based Markov models \citep{Chen2006}. We focus on the estimation of cumulative distribution functions. Following \cite{bortot2000}, however, one could also use kernel methods to estimate their densities.

The structure of the paper is as follows. Tail processes are reviewed in Section~\ref{sec:prelimin}, with special attention to those of Markov chains. The estimators of the tail process of a regularly Markov chain are described in Section~\ref{sec:estim}. Their asymptotic properties are worked out in Sections~\ref{sec:asym}, whereas their finite-sample performance is evaluated via Monte Carlo simulations in Section~\ref{sec:simul} involving models presented in the Section~\ref{sec:examples}. In Section~\ref{sec:applic}, the estimators are applied to analyze time series of daily log returns of Google and UBS stock prices, revealing interesting patterns regarding the succession of large losses and gains. Proofs and calculations are deferred to Section~\ref{appendix}.

Some notational conventions: the law of a random object is denoted by $\law{\,\cdot\,}$. Weak convergence is denoted by the arrow $\dto$. The indicator variable of the event $E$ is denoted by $\1(E)$. The set of integers is denoted by $\ZZ$, while $\NN = \{h \in \ZZ : h \ge 1\}$.

\section{Tail processes of Markov chains}
\label{sec:prelimin}
A strictly stationary time series $(X_{t})_{t\in \ZZ}$ is said to have a \emph{tail process} $(Y_{t})_{t\in \ZZ}$ if, for all $s,t\in\ZZ$ such that $s\leq t$,
\begin{equation}
\label{eq:tailprocess}
  \law{ u^{-1} X_s,\ldots,u^{-1} X_t \mid \abs{X_0} > u } \dto \law{ Y_s,\ldots,Y_t }, \qquad u \to \infty,
\end{equation}
with the implicit understanding that the law of $\abs{Y_0}$ is non-degenerate.

Specializing equation~\eqref{eq:tailprocess} to $t = 0$ implies that $\prob{ \abs{X_0} > uy } / \prob{ \abs{X_0} > u } \to \prob{ \abs{Y_0} > y }$ as $u \to \infty$ for all continuity points $y$ of the law of $|Y_0|$. Since the law of $\abs{ Y_0 }$ was supposed to be non-degenerate, it follows that the function $u \mapsto \prob{ \abs{X_0} > u }$ is regularly varying at infinity: there exists $\alpha > 0$ such that
\begin{equation}
\label{eq:RV}
  \lim_{u \to \infty} \frac{\prob{ \abs{X_0} > uy }}{\prob{ \abs{X_0} > u }}
  = y^{-\alpha}, \qquad y \in (0, \infty).
\end{equation}
The law of $\abs{Y_0}$ is thus Pareto($\alpha$), i.e., $\prob{ \abs{Y_0} > y } = y^{-\alpha}$ for all $y \ge 1$. More generally, by \citet[Theorem~2.1]{basrak:segers:2009}, the time series $(X_t)_{t \in \ZZ}$ admits a tail process $(Y_t)_{t \in \ZZ}$ with non-degenerate $\abs{Y_0}$ if and only if $(X_t)_{t \in \ZZ}$ is jointly regularly varying with some index $\alpha > 0$, i.e., if for all integers $k \leq l$ the random vector $(X_{k},\ldots,X_{l})$ is multivariate regularly varying with index $\alpha$.

Many time series models are jointly regularly varying and hence admit a tail process. Examples include linear processes with heavy-tailed innovations, solutions to stochastic recurrence equations, and models of the ARCH and GARCH families. Sufficient conditions for such models to be regularly varying can be found in \cite{Davis:Mikosch:Zhao:2013}.

The \emph{spectral tail process} is defined by
\[
  \Theta_t = Y_t / \abs{Y_0}, \qquad t \in \ZZ.
\]
By \eqref{eq:tailprocess} and the continuous mapping theorem, it follows that for all $s,t\in\ZZ$ such that $s\leq t$
\begin{equation}
\label{eq:spectralprocess}
   \mathcal{L}\left(X_s / \abs{X_0},\ldots,X_t / \abs{X_0} \mid \abs{X_0} > u \right)
  \rightarrow  \mathcal{L}\left( \Theta_s,\dots, \Theta_t \right), \qquad u \to \infty.
\end{equation}
The difference between \eqref{eq:tailprocess} and \eqref{eq:spectralprocess} is that in the latter equation, the variables $X_t$ are normalized by $\abs{X_0}$ rather than by the threshold $u$. Such auto-normalization allows the tail process to be decomposed into two stochastically independent components, i.e.,
\begin{equation}\label{eq:decomp}
  Y_t = \abs{Y_0} \, \Theta_t, \qquad t \in \ZZ.
\end{equation}
Independence of $\abs{Y_0}$ and $(\Theta_t)_{t \in \ZZ}$ is stated in \citet[Theorem~3.1]{basrak:segers:2009}. The random variable $\abs{Y_0}$ characterizes the magnitudes of extremes, whereas $(\Theta_t)_{t \in \ZZ}$ captures serial dependence. The spectral tail process at time $t = 0$ yields information on the relative weights of the upper and lower tails of $\abs{X_0}$: since $\Theta_0 = Y_0 / \abs{Y_0} = \sign(Y_0)$, we have
\begin{align}
\label{p_prob}
  p &= \prob{ \Theta_0 = +1 } = \lim_{u \to \infty} \frac{\prob{ X_0 > u }}{\prob{ \abs{X_0} > u }}, &
  1-p &= \prob{ \Theta_0 = -1 }.
\end{align}

The distributions of the \emph{forward} tail process $(Y_t)_{t \ge 0}$ and the \emph{backward} tail process $(Y_t)_{t \le 0}$ mutually determine each other. The precise connection between the forward and backward (spectral) tail processes is captured by Theorem~3.1 in \cite{basrak:segers:2009}. For all $i,s,t\in\ZZ$ with $s\leq 0\leq t$ and for all measurable functions $f:\reals^{t-s+1} \to \reals$ satisfying $f(y_s,\ldots,y_t)=0$ whenever $y_0=0$, we have, provided the expectations exist,
\begin{equation}
\label{eq:timechange}
\expec{ f\left(\Theta_{s-i},\ldots,\Theta_{t-i}\right) } =
\expec{  f\left(\frac{\Theta_{s}}{\vert\Theta_{i}\vert},\ldots,\frac{\Theta_{t}}{\vert\Theta_{i}\vert}\right) \, \vert\Theta_{i}\vert^\alpha \, \1\{ \Theta_i \ne 0 \} }.
\end{equation}
We will refer to \eqref{eq:timechange} as the \emph{time-change formula}. By exploiting the time-change formula, we will be able to improve upon the efficiency of estimators of the tail process.

A common procedure in multivariate extreme value theory is to standardize the margins. For jointly regularly varying time series, such a standardization is possible too, although some care is needed because of the possible presence of both positive and negative extremes.

\begin{lemma}
\label{lem:standardization}
Let $(X_t)_{t \in \ZZ}$ be a stationary time series, jointly regularly varying with index $\alpha > 0$, and having spectral tail process $(\Theta_t)_{t \in \ZZ}$. Put $\bar{F}_{\abs{X_0}}(u) = \prob{ \abs{X_0} > u }$ for $u \ge 0$. Define a stationary time series $(X_t^*)_{t \in \ZZ}$ by
\begin{equation}
\label{eq:standardization}
  X_t^* = \frac{\sign(X_t)}{\bar{F}_{\abs{X_0}}(\abs{X_t})}, \qquad t \in \ZZ.
\end{equation}
Then $(X_t^*)_{t \in \ZZ}$ is jointly regularly varying with index $1$. Its spectral tail process $(\Theta_t^*)_{t \in \ZZ}$ is given by
\begin{equation}
\label{eq:standardization:theta}
  \Theta_t^* = \sign(\Theta_t) \, \abs{\Theta_t}^\alpha, \qquad t \in \ZZ.
\end{equation}
\end{lemma}

In \eqref{eq:standardization:theta}, note that the map $y \mapsto \sign(y) \, \abs{y}^\alpha$ is monotone and symmetric.
The standardized series $(X_t^*)_{t \in \ZZ}$ may be regularly varying even if the original series $(X_t)_{t \in \ZZ}$ is not. In that sense, the standardization procedure in \eqref{eq:standardization} widens the field of possible applications of tail processes. For instance, the marginal distributions of enviromental variables are often light-tailed rather than regularly varying. After standardization as in Lemma~\ref{lem:standardization}, the serial dependence between of extremes of such time series may still be modelled via tail processes.

Some time series models exhibit asymptotic independence of consecutive observations, that is, $\prob{ \abs{X_k} > u \mid \abs{X_0} > u } \to 0$ as $\prob{ \abs{X_0} > u } \to 0$ for all $k\in\ZZ\setminus\{0\}$. Well-known examples are non-degenerate Gaussian time series and classical stochastic volatility models. In such cases, the spectral tail process is noninformative in the sense that $\Theta_k = 0$ almost surely for all $k \ne 0$. More refined approaches to handle tail independence were developed in \cite{ledfordtawn96,ledfordtawn03} and, more recently, in \cite{janssendrees13} and \cite{kuliksoulier13}.

\subsection*{Regularly varying Markov chains}

For the purpose of statistical inference, the class of spectral tail processes is too large to be really useful: without additional modelling assumptions, it is impossible to estimate all limiting finite-dimensional distributions that appear in \eqref{eq:tailprocess} or \eqref{eq:spectralprocess}. Therefore, it is reasonable to consider families of spectral tail processes arising under additional constraints on the underlying time series.

One such family was identified in \cite{segersmarkov} and \cite{segersjanssen} in the context of first-order Markov chains. Let $(\Theta_t)_{t \in \ZZ}$ be the spectral tail process of an $\alpha$-regularly varying, stationary time series $(X_t)_{t \in \ZZ}$, arising as the limit process in \eqref{eq:spectralprocess}. Put $p = \prob{ \Theta_0 = 1 }$ as in \eqref{p_prob}.
 Introduce random variables $A_1, B_1, A_{-1}, B_{-1}$ (or rather their laws) as follows: if $p > 0$, then, as $u \to \infty$,
\begin{align}
\label{A1_prob}
  \mathcal{L}(X_1/X_0 \mid X_0>u) &\dto \mathcal{L}(A_{1})=\law{\Theta_1\mid\Theta_0=+1}, \\
\label{A-1_prob}
  \mathcal{L}(X_{-1}/X_0 \mid X_0>u) &\dto \mathcal{L}(A_{-1})=\law{\Theta_{-1}\mid\Theta_0=+1},
\end{align}
and if $p < 1$, then
\begin{align}
\label{B1_prob}
  \mathcal{L}(X_1/X_0 \mid X_0<-u) &\dto \mathcal{L}(B_{1})=\law{-\Theta_1\mid\Theta_0=-1}, \\
\label{B-1_prob}
  \mathcal{L}(X_{-1}/X_0 \mid X_0<-u) &\dto \mathcal{L}(B_{-1})=\law{-\Theta_{-1}\mid\Theta_0=-1}.
\end{align}
Further, let $\Theta_0,A_1,A_{-1},A_2,A_{-2},\ldots, B_1,B_{-1},B_2,B_{-2}$ be independent random variables such that $\law{A_t} = \law{A_1}$, $\law{A_{-t}} = \law{A_{-1}}$,  $\law{B_t} = \law{B_1}$, and $\law{B_{-t}} = \law{B_{-1}}$ for all  $t \in\NN$. Then the spectral tail process $(\Theta_t)_{t \in \ZZ}$ is said to be a \emph{Markov spectral tail chain} if the following holds: the forward spectral tail process is given recursively by
\begin{equation}
\label{eq:Markov:forward}
  \Theta_t =
  \begin{cases}
  \Theta_{t-1}A_t & \text{if $\Theta_{t-1}>0$}, \\
  0 & \text{if $\Theta_{t-1}=0$}, \\
  \Theta_{t-1}B_t & \text{if $\Theta_{t-1}<0$},
  \end{cases}
  \qquad t\in\NN,
\end{equation}
whereas the backward spectral tail process is given by
\begin{equation}
\label{eq:Markov:backward}
  \Theta_{-t} =
  \begin{cases}
    \Theta_{-t+1}A_{-t} & \text{if $\Theta_{-t+1}>0$}, \\
    0 & \text{if $\Theta_{-t+1}=0$}, \\
    \Theta_{-t+1}B_{-t} & \text{if $\Theta_{-t+1}<0$},
  \end{cases}
  \qquad t\in\NN.
\end{equation}
If $p = 1$, then $\Theta_t \ge 0$ almost surely for all $t \in \ZZ$ and thus the definition of $B_{\pm t}$ is immaterial; similarly if $p = 0$. This can be seen by applying the time-change formula \eqref{eq:timechange}.

The motivation behind the above definition is that such spectral tail processes typically arise when $(X_t)_{t \in \ZZ}$ is a stationary, first-order Markov chain; see Theorem~5.2 in \cite{segersmarkov} and Corollary~5.1 in \cite{segersjanssen}. However, they may as well arise in settings where the underlying process, $(X_t)_{t \in \ZZ}$, is non-Markovian; see Remark~5.1 in \cite{segersjanssen}. The forward and backward spectral tail processes $(\Theta_t)_{t \ge 0}$ and $(\Theta_t)_{t \le 0}$ are Markovian themselves, and, conditionally on $\Theta_0$, they are independent. Their structure is that of a geometric random walk where the distribution of the increment at time $t$ depends on the sign of the process at time $t-1$. The point zero acts as an absorbing state.

For Markov spectral tail chains, the distribution of the forward part $(\Theta_t)_{t \ge 0}$ is determined by $p$, $A_1$, and $B_1$. Given additionally the index of regular variation $\alpha > 0$, the distributions of $A_{-1}$ and $B_{-1}$ and thus of the backward part $(\Theta_t)_{t \le 0}$ can be reconstructed from the time-change formula \eqref{eq:timechange}; see Lemma~\ref{lem:tc:AB} below. It follows that the law of a Markov spectral tail process is determined by $\alpha > 0$, $p \in [0, 1]$, and the laws of $A_1$ and $B_1$. This reduction provides a handle on the spectral tail process that can be exploited for statistical inference.

\section{Estimating Markov spectral tail processes}
\label{sec:estim}

In this section we propose estimators for $p$, $A_1$ and $B_1$. In combination with the index of regular variation $\alpha > 0$, this triplet fully determines the law of a Markov spectral tail process as defined in equations~\eqref{eq:Markov:forward} and \eqref{eq:Markov:backward}, and of the tail processes $(Y_t)_{t\in\ZZ}$.

Replacing population distributions by sampling distributions in the left-hand sides of \eqref{A1_prob} and \eqref{B1_prob} yields forward estimators for the laws of $A_1$ and $B_1$. However, exploiting the time-change formula \eqref{eq:timechange} allows to express the laws of $A_1$ and $B_1$ in terms of $A_{-1}$ and $B_{-1}$ (and $p$ and $\alpha$). These expressions motivate so-called backward estimators for $A_1$ and $B_1$. Convex combinations of forward and backward estimators finally produce mixture estimators. For an appropriate choice of the mixture weights, the mixture estimators may be more efficient than both the forward and the backward estimators separately.

In order to estimate $p = \prob{\Theta_0 = 1}$, we simply take the empirical version of \eqref{p_prob}, yielding
\begin{equation}\label{p_est}
\hat{p}_n = \frac{\sum_{i=1}^{n} \1\left(X_i > u_n\right)}{\sum_{i=1}^{n}\1\left(\abs{X_i} > u_n\right)}.
\end{equation}
For $\hat{p}_n$ to be consistent and asymptotically normal, the threshold sequence $u_n$ should tend to infinity at a certain rate described in detail in condition {\bf (B)} in the next section.

For estimating the cdf, $F^{(A_1)}$, of $A_1$ we propose
\begin{equation}
\label{forward_A1}
  \CDFfA{x}
  =
  \frac{\sum_{i=1}^{n}\1\left(X_{i+1}/X_i \leq x,\; X_i>u_n\right)}{\sum_{i=1}^{n}\1\left( X_i > u_n\right)},
\end{equation}
which we refer to as the \emph{forward estimator} of the cdf of $A_1$. Similarly, for the forward estimator of the cdf of $B_1$ we take
\begin{equation}
\label{forward_B1}
  \CDFfB{x}
  =
  \frac{\sum_{i=1}^{n}\1\left(X_{i+1}/X_i \leq x,\; X_i< - u_n\right)}{\sum_{i=1}^{n}\1\left( X_i < - u_n\right)}.
\end{equation}
The forward estimators of the cdf's of $A_1$ and $B_1$ are empirical versions of the left-hand sides of \eqref{A1_prob} and \eqref{B1_prob}, respectively. Note that one can expect consistency of these estimators only if the target distribution functions are continuous in $x$, because otherwise $\prob{X_1/X_0\le x \mid X_0>u_n}$ need not converge to $\prob{A_1 \le x}$, for instance.

The time-change formula~\eqref{eq:timechange} yields a different representation of $A_1$ and $B_1$, motivating different estimators than the ones above, based on different data points. For  ease of reference, we record the relevant formulas in a lemma, whose proof is given in Appendix~\ref{sec:proofs}.

\begin{lemma}
\label{lem:tc:AB}
Let $(X_t)_{t \in \ZZ}$ be a stationary time series, jointly regularly varying with index $\alpha$ and spectral tail process $(\Theta_t)_{t \in \ZZ}$. Let $A_1, A_{-1}, B_1, B_{-1}$ be given as in \eqref{A1_prob} to \eqref{B-1_prob}. If $p = \prob{ \Theta_0 = 1 } > 0$, then
\begin{align}
\label{eq:A1:tc:pos}
  \prob{ A_1 > x }
  &=
  \expec{
    A_{-1}^{\alpha} \,
    \1\left(1/A_{-1} > x\right)
  }, && x \ge 0, \\
\label{eq:A1:tc:neg}
  \prob{ A_1 \le x }
  &= \frac{1-p}{p}
  \expec{
    \left(-B_{-1}\right)^{\alpha} \,
    \1\left(1/B_{-1} \leq x\right)
  }, && x < 0.
\end{align}
Similarly, if $p < 1$, then
\begin{align}
\label{eq:B1:tc:pos}
  \prob{ B_1 > x }
  &=
  \expec{
    B_{-1}^\alpha \,
    \1(1/B_{-1} > x)
  }, && x \ge 0, \\
\label{eq:B1:tc:neg}
  \prob{ B_1 \le x }
  &= \frac{p}{1-p}
  \expec{
    (-A_{-1})^\alpha \,
    \1(1/A_{-1} \le x)
  }, && x < 0.
\end{align}
Formulas \eqref{eq:A1:tc:pos} to \eqref{eq:B1:tc:neg} remain valid when the time instances $1$ and $-1$ are interchanged.
\end{lemma}

Assume for the moment that $\alpha$ is known. Below, we will consider the more realistic situation that $\alpha$ is unknown. Lemma~\ref{lem:tc:AB} suggests the following \emph{backward estimator} of the cdf of $A_1$:
\begin{equation}
\label{backward_A1}
\CDFbA{x} =
\begin{cases}
1-\dfrac{\sum_{i=1}^{n}\left( \frac{X_{i-1}}{X_i}\right)^{\alpha}\1\left(X_i/ X_{i-1} >x,
\; X_i > u_n \right)}{\sum_{i=1}^{n}\1\left(  X_i > u_n\right)} &
\text{if $x \geq 0$}, \\[1em]
\dfrac{\sum_{i=1}^{n}\left( \frac{-X_{i-1}}{X_i}\right)^{\alpha}\1\left(X_i/ X_{i-1} \leq x,
\; X_i < -u_n \right)}{\sum_{i=1}^{n}\1\left(  X_i > u_n\right)} &
\text{if $x<0$.}
\end{cases}
\end{equation}
Similarly, we define the backward estimator of the cdf of $B_1$ as
\begin{equation}
\label{backward_B1}
\CDFbB{x} =
\begin{cases}
1-\dfrac{\sum_{i=1}^{n}\left( \frac{X_{i-1}}{X_i}\right)^{\alpha}\1\left(X_i/ X_{i-1} >x,
\; X_i < -u_n \right)}{\sum_{i=1}^{n}\1\left(  X_i < -u_n\right)} &
\text{if $x\geq 0$}, \\[1em]
\dfrac{\sum_{i=1}^{n}\left( \frac{-X_{i-1}}{X_i}\right)^{\alpha}\1\left(X_i/ X_{i-1} \leq x,
\; X_i > u_n \right)}{\sum_{i=1}^{n}\1\left(  X_i < -u_n\right)} &
\text{if $x<0$}.
\end{cases}
\end{equation}

For $\abs{x}$ large, the backward estimators usually have a smaller variance than the forward estimators. To see this, note that for negative $x$ with large modulus only very few summands in the numerator of \eqref{forward_A1} do not vanish, because $X_{i+1}$ must be even larger in absolute value than $|x|X_i>|x|u_n$, leading to a large variance of the numerator. In contrast, usually many more non-vanishing terms will be summed up in the numerator of \eqref{backward_A1}, while each of them gets a rather low weight $(-X_{i-1}/X_i)^\alpha\le |x|^{-\alpha}$, leading to a smaller variance. For large positive $x$ one may argue similarly by considering the corresponding estimators of the survival function. Indeed, we show in Remark~\ref{rem:xgeq1} that, provided $x \geq 1$, the backward estimator of the cdf of $A_1$ at $x$ has a smaller asymptotic variance than the forward estimator.

For well-chosen weights, convex combinations of the forward and backward estimators can achieve a lower asymptotic variance than each of the estimators individually. Unfortunately, the expression for the asymptotic covariance of the two estimators is intractable; see Corollary \ref{cor:asnorm}. It remains an open issue how to choose the mixture weights in order to minimize the asymptotic variance.

A pragmatic approach is to give more weight to the forward estimator for  small $\abs{x}$ and to give more weight to the backward estimator for large $\abs{x}$. To this end, define weights by
\[\lambda(x)=\text{max}\left(1-\abs{x},\;0\right).\]
The \emph{mixture estimator} for the cdf of $A_1$ is defined as
\begin{equation}
\label{mixture_A1}
\CDFmA{x} =
\begin{cases}
\lambda(x)\CDFfA{x} +\left[  1-\lambda(x)\right] \CDFbA{x} & \text{if $x \geq 0$,} \\[1ex]
\lambda(x)\CDFfA{x} +\left[ 1-\lambda(x)\right] \CDFbA{x} & \text{if $x <  0$.}
\end{cases}
\end{equation}
The mixture estimator for $B_1$ is defined by replacing $A_1$ in \eqref{mixture_A1} with $B_1$.

The backward and the mixture estimators require the value of the index $\alpha$ of regular variation, which is unknown in most applications. There are at least two approaches to deal with this issue:
\begin{enumerate}
\item
Estimate $\alpha$ separately, for instance, by the Hill--type estimator
\begin{equation}
\label{Hill}
  \hat{\alpha} = \frac{\sum_{i=1}^{n} \1\left( \abs{ X_i } > u_n\right)}{\sum_{i=1}^{n} \log \left(\abs{ X_{i} }/ u_n\right) \1\left(\abs{ X_i } > u_n\right)},
\end{equation}
and plug in the estimated value of $\alpha$ in \eqref{backward_A1}, \eqref{backward_B1} and \eqref{mixture_A1}.
\item
Employ an empirical version of the transformation in Lemma~\ref{lem:standardization} to ensure that, after transformation, $\alpha=1$.  The transformation in \eqref{eq:standardization} requires the tail function $\bar{F}_{\abs{X_0}}(u) = \prob{ \abs{X_0} > u }$. This function can be estimated, for instance, by
\begin{equation}
\label{emp_trans}
 \hat{\bar{F}}_{\abs{X_0},n}(u) = 1-\frac{1}{n+1} \sum_{j=1}^n \1(\abs{X_j} \leq u),
\end{equation}
where we divide by $n + 1$ rather than by $n$ in order to avoid division by zero later on. The transformed variable
\[ \hat{X}_{n,i}^* = \sign(X_i) / \hat{\bar{F}}_{\abs{X_0},n}(\abs{X_i}) \] is based on the sign of $X_i$ and the rank of $\abs{X_i}$ among $\abs{X_1}, \ldots, \abs{X_n}$.
\end{enumerate}

In the simulation study in Section~\ref{sec:simul}, the mixture estimator based on the rank-transformed data performs better than the plug-in version. Note, however, that the two approaches are not directly comparable: with the second approach, what we estimate is the tail process $(\Theta_t^*)_{t \in \ZZ}$ of the transformed series $(X_t^*)_{t \in \ZZ}$. From \eqref{eq:standardization:theta} and \eqref{A1_prob}, it follows that, if $(\Theta_t)_{t \in \ZZ}$ is a Markov spectral tail chain as in \eqref{eq:Markov:forward} and \eqref{eq:Markov:backward}, then so is $(\Theta_t^*)_{t \in \ZZ}$, with $A_t$ and $B_t$ to be replaced by $A_t^* = \sign(A_t) \, \abs{A_t}^\alpha$ and $B_t^* = \sign(B_t) \, \abs{B_t}^\alpha$, respectively. Combining the above two estimation approaches, one could even recover $A_t$ and $B_t$ via $A_t = \sign(A_t^*) \, \abs{A_t^*}^{1/\alpha}$ and $B_t = \sign(B_t^*) \, \abs{B_t^*}^{1/\alpha}$. Finally, as $(X_t^*)_{t\in\ZZ}$ may have a tail process even when the original time series has none, the second approach is more widely applicable.

\section{Large sample theory}
\label{sec:asym}

Under certain conditions, the standardized estimation errors of the forward and the backward estimators converge jointly to a centered Gaussian process. In order not to overload the presentation, we focus on nonnegative Markov chains. In that case, the distribution of $\Theta_1=A_1$ determines the distribution of the forward spectral tail process, and thus, via the time-change formula, together with $\alpha$, also the one of the backward spectral tail process. We distinguish between the cases where $\alpha$ is known (Section~\ref{sec:CLT:known}) and unknown (Section~\ref{sec:CLT:unknown}). In addition, we briefly indicate how the conditions and results must be modified in the real-valued case (Remark~\ref{rem:posneg}).

\subsection{Known index of regular variation}
\label{sec:CLT:known}

If the index of regular variation, $\alpha$, is known, all estimators under consideration can be expressed in terms of \emph{generalized tail array sums}, that is, statistics of the form $\sum_{i=1}^n \phi(X_{n,i})$, with
\begin{equation} \label{eq:Xnidef}
  X_{n,i} := \frac{(X_{i-1},X_i,X_{i+1})}{u_n} \1(X_i>u_n).
\end{equation}
\cite{drees2010limit} give conditions under which, after standardization, such statistics converge to a centered Gaussian process, uniformly over appropriate families of functions $\phi$. From these results we will deduce a functional central limit theorem for the processes of forward and backward estimators defined in \eqref{forward_A1} and \eqref{backward_A1}, respectively.

To ensure consistency, the threshold $u_n$ must tend to infinity such that
$$ v_n:=\prob{X_0>u_n}
$$
tends to 0, but the expected number $nv_n$ of exceedances tends to infinity.
Moreover, let
\[
  \beta_{n,k} := \sup_{1\le l\le n-k-1}
  \expec{\sup_{B\in\mathcal{B}_{n,l+k+1}^n} \abs{\prob{B \mid \mathcal{B}_{n,1}^l}-\prob{B}}}
\]
denote the $\beta$-mixing coefficients. Here $\mathcal{B}_{n,i}^j$ is the $\sigma$-field generated by $(X_{n,l})_{i\le l\le j}$. We assume that there exist sequences $l_n,r_n\to\infty$ and some $x_0\ge 0$ such that the following conditions hold:

\begin{description}
\item[(A($x_0$))]
The cdf, $F^{(A_1)}$, of $A_1$ is continuous on $[x_0,\infty)$.
\item[(B)]
\begin{enumerate}[(i)]
\item
As $n \to \infty$, we have $l_n\to\infty$, $l_n=o(r_n)$, $r_n=o((n v_n)^{1/2})$, $r_nv_n\to 0$;
\item
$\beta_{n,l_n}n/r_n\to 0$ as $n \to \infty$ and $\lim_{m\to\infty} \limsup_{n\to\infty} \beta_{n,m} = 0$.
\end{enumerate}
\end{description}
Condition (B) pose restrictions on the rate at which $v_n$ tends to 0 and thus on the rate at which $u_n$ tends to $\infty$.
Sufficient conditions to ensure that a Markov chain is $\beta$-mixing can be found in \citet[Section~2.4]{doukhan95}. Usually, $\beta_{n,k}=O(\eta^k)$ for some $\eta\in(0,1)$ and one may choose $l_n=O(\log n)$, and (B) is fulfilled for a suitably chosen $r_n$ if $(\log n)^2/n=o(v_n)$ and $v_n =o(1/\log n)$.
\begin{description}
\item[(C)]
For all $k\in\{0,\ldots,r_n\}$ there exists
\[
  s_n(k) \ge \prob{X_k>u_n\mid X_0>u_n}
\]
such that $\lim_{n\to\infty} \sum_{k=1}^{r_n} s_n(k) = \sum_{k=1}^\infty \lim_{n\to\infty} s_n(k)<\infty$.
\end{description}
Typically $s_n(k)$ will be of the form $b_n+c_k$ with $b_n=o(1/r_n)$ and $\sum_{k=1}^\infty c_k<\infty$. The interchangeability of the limit and the sum is then automatically fulfilled. For stochastic recurrence equations (Section~\ref{sec:examples:sre}), conditions~(B) and~(C) are verified in Example~\ref{exam:SRE} below.

Under these conditions, one can prove the asymptotic normality of relevant generalized tail array sums (see Proposition \ref{th:procconv} below) and thus the joint asymptotic normality of the forward and the backward estimator of $ F^{(A_1)}$ centered by their expectation. However, additional conditions are needed to ensure that their bias is asymptotically negligible:
\begin{eqnarray}
  \sup_{x\in [x_0,\infty)} \left\lvert \prob {\frac{X_1}{X_0} \le x\,\Big|\, X_0>u_n} -  F^{(A_1)}(x) \right\rvert
  & = & o\bigl((nv_n)^{-1/2}\bigr),  \label{eq:bias1}\\
  \sup_{y\in [y_0,\infty)} \left\lvert \expec{\Big(\frac{X_{-1}}{X_0}\Big)^\alpha \1(X_0/X_1>y)\,\Big|\, X_0>u_n} - \bar F^{(A_1)}(y) \right\rvert
  & = & o\bigl((nv_n)^{-1/2}\bigr). \label{eq:bias2}
\end{eqnarray}
Here $\bar F^{(A_1)}:=1-F^{(A_1)}$ denotes the survival function of $A_1$ (and hence of $\Theta_1$).
These conditions are fulfilled if $nv_n$ tends to $\infty$ sufficiently slowly, because by definition of the spectral tail process and by \eqref{eq:A1:tc:pos}, the left-hand sides in \eqref{eq:bias1}--\eqref{eq:bias2} tend to $0$ if $F^{(A_1)}$ is continuous on $[x_0,\infty)$.

\begin{theorem}
\label{cor:asnorm}
Let $(X_t)_{t\in\ZZ}$ be a stationary, regularly varying process with a Markov spectral tail chain. If (A($x_0$)), (B),  (C), \eqref{eq:bias1} and \eqref{eq:bias2}  are fulfilled for some $x_0\ge 0$ and $y_0\in [x_0,\infty)\cap (0,\infty)$, then
\begin{equation}
\label{eq:estconv}
  (nv_n)^{1/2}
  \begin{pmatrix}
    (\CDFfA{x}- F^{(A_1)}(x))_{x\in[x_0,\infty)} \\
    (\CDFbA{y}- F^{(A_1)}(y))_{y\in[y_0,\infty)}
  \end{pmatrix}
  \dto
  \begin{pmatrix}
    (Z^{(f,A_1)}(x))_{x\in[x_0,\infty)} \\
    (Z^{(b,A_1)}(y))_{y\in[y_0,\infty)}
  \end{pmatrix},
  \qquad n \to \infty,
\end{equation}
where the limit is a centered Gaussian process whose covariance function is given by
\begin{eqnarray*}
  \cov{Z^{(f,A_1)}(x), \, Z^{(f,A_1)}(y)} & = &  \bar F^{(A_1)}(\max (x,y))-\bar F^{(A_1)}(x)\bar F^{(A_1)}(y),\\
  \cov{Z^{(b,A_1)}(x), \, Z^{(b,A_1)}(y)} & = &  \expec{\Theta_1^{-1}\1(\Theta_1>\max(x,y))}-\bar F^{(A_1)}(x)\bar F^{(A_1)}(y),
\end{eqnarray*}
and
\begin{eqnarray*}
  \cov{Z^{(f,A_1)}(x), \, Z^{(b,A_1)}(y)} &=&
  \sum_{k=1}^\infty {(\Theta_{k-1}/\Theta_k)^\alpha
  \1(\Theta_1>x,\Theta_k/\Theta_{k-1}>y,Y_k>1)} \\
  && \mbox{} - \bar F^{(A_1)}(x) \sum_{k=1}^\infty \expec{(\Theta_{k-1}/\Theta_k)^\alpha\1(\Theta_k/\Theta_{k-1}>y,Y_k>1)}\\
  && \mbox{} - \bar F^{(A_1)}(y) \sum_{k=1}^\infty \prob{\Theta_1>x,Y_k>1}\\
  && \mbox{} + \bar F^{(A_1)}(x) \, \bar F^{(A_1)}(y) \sum_{k=1}^\infty \prob{Y_k>1}.
\end{eqnarray*}
\end{theorem}

\begin{remark}
\label{rem:xgeq1}
For $x \ge 1$, we have
\begin{eqnarray*}
  \var(Z^{(b,A_1)}(x))
  & = & \expec{ \Theta_1^{-1}\mathbf{1}(\Theta_1>x)}-\big(\bar F^{(A_1)}(x)\big)^2\\
  & < & \prob{\Theta_1>x}-\big(\bar F^{(A_1)}(x)\big)^2\\
  & = & \var(Z^{(f,A_1)}(x)),
\end{eqnarray*}
provided $ \bar F^{(A_1)}(x)=\prob{A_1>x}=\prob{\Theta_1>x}>0$. Hence, for such $x$, when the tail index $\alpha$ is known, the backward estimator is asymptotically more efficient than the forward estimator.
\end{remark}

\begin{remark}
\label{rem:cont}
  While it is not too restrictive to assume that the cdf of $A_1$ is continuous on $(0,\infty)$, often $\mathcal{L}(A_1)$ has positive mass at 0; see Example~\ref{ex:tEVcopula}. In this case, one may prove a version of Theorem~\ref{cor:asnorm} where the first coordinate in \eqref{eq:estconv} is replaced with
  \[
    (nv_n)^{1/2} \left( w(x) \bigl( \CDFfA{x}-F^{(A_1)}(x) \bigr) \right)_{x\in[0,\infty)}
    \dto \left( w(x) \, Z^{(f,A_1)}(x) \right)_{x\in[0,\infty)}, \qquad n \to \infty,
  \]
  for a weight function $w(x)=h \bigl( F^{(A_1)}(x)-F^{(A_1)}(0) \bigr)$ and any nondecreasing, continuous function $h$ with $h(0)=0$.
\end{remark}

\begin{remark}\label{rem:markov}
  Note that a similar result also holds true without assuming the  Markovianity of the spectral tail chain, but then the formulas for the covariance function of the limiting process are more involved. The simple explicit formulas for the asymptotic variances obtained in Theorem \ref{cor:asnorm} can be used to construct pointwise confidence intervals for the cdfs of $A_1$ or $B_1$ by a plug-in approach. However, if one wants to derive uniform confidence bands or tests for these cdfs then a resampling procedure may be advisable. The same holds true if one cannot assume that the tail sequence has the Markov property. The analysis of such methods is left for future work.

\end{remark}

\subsection{Unknown index of regular variation}
\label{sec:CLT:unknown}

In most applications, the index of regular variation, $\alpha$, is unknown. In the definition of the backward estimator $\CDFbA{y}$, it must then be replaced with a suitable estimator. A popular estimator of $\alpha$ is the Hill-type estimator \eqref{Hill}. More generally, one may consider estimators that can be written in the form
\begin{equation} \label{eq:alphaest}
  \hat \alpha_n = \frac{\sum_{i=1}^n \1(X_i>u_n)}{\sum_{i=1}^n \tilde \psi(X_i/u_n)\1(X_i>u_n)+R_n}
\end{equation}
with a remainder term $R_n=o_P((nv_n)^{1/2})$ and a suitable function $\tilde \psi:[0,\infty)\to[0,\infty)$ which is a.s.\ continuous w.r.t.\ $\law{Y_k}$ for all $k\in\NN \cup \{0\}$ such that $\texpec{\tilde\psi(Y_0)}=1/\alpha$ and $\tilde\psi(x)=0$ for all $x\in[0,1]$. Obviously, the Hill-type estimator is of this form with $\tilde\psi(x) = \log(x) \, \1(x>1)$. Under weak dependence conditions, other well-known estimators like the maximum likelihood estimator in a generalized Pareto model examined by \cite{smith87} and the moment estimator suggested by \cite{dekketal89} can be written in this way too; see \citet[Example~4.1]{drees98b} and \citet[Example~4.1]{drees98a} for similar results in the case of i.i.d.\ sequences.

Estimators of type \eqref{eq:alphaest} can be approximated by the ratio of the generalized tail array sums corresponding to the functions $\phi_1(y_{-1},y_0,y_1):= \1(y_0>1)$ and $\psi(y_{-1},y_0,y_1):=\tilde\psi(y_0)\1(y_0>1)$, respectively, and their asymptotic behavior can hence be derived from Theorem 2.3 of \cite{drees2010limit}. To this end, we replace (C) with the following condition:
\begin{description}
\item[(C')]
For all $0\le k\le r_n$ there exists
\begin{equation} \label{eq:snkdef2}
s_n(k)\ge \expec{\max\Big(\Big|\tilde\psi\Big(\frac{X_0}{u_n}\Big)\Big|, \1(X_0>u_n)\Big) \max\Big(\Big|\tilde\psi\Big(\frac{X_k}{u_n}\Big)\Big|,\1(X_k>u_n)\Big) \, \Big| \, X_0>u_n}
\end{equation}
such that $\lim_{n\to\infty} \sum_{k=1}^{r_n} s_n(k)=\sum_{k=1}^\infty \lim_{n\to\infty} s_n(k)<\infty$.

Moreover,  there exists $\delta>0$ such that
\begin{equation}
\label{eq:psibdd}
  \sum_{k=1}^{r_n}
  \bigg(
    \expec{
      \Bigl| \tilde\psi \Bigl( \frac{X_0}{u_n}\Big)\tilde\psi\Big(\frac{X_k}{u_n} \Bigr) \Bigr|^{1+\delta}
      \, \Big| \,
      X_0 > u_n
    }
  \bigg)^{1/(1+\delta)}
  = O(1),
  \qquad n \to \infty.
\end{equation}
\end{description}

If $\tilde\psi$ is bounded, then (C') follows from condition (C), but in general it is more restrictive, though it can often be established by similar arguments. For example, condition (C') holds for the solutions to the stochastic recurrence equation studied in Example~\ref{exam:SRE} and the Hill estimator, i.e., for $\tilde\psi(x) = \log (x) \, \1(x>1)$.

The following result gives the asymptotic normality of $\hat\alpha_n$ centered at
\begin{equation}
\label{eq:alphan}
  \alpha_n := 1 / \texpec{\tilde\psi(X_0/u_n)\mid X_0>u_n}.
\end{equation}
This quantity tends to $\alpha$ as $n \to \infty$ by the assumptions on the function $\tilde\psi$ and condition (C').

\begin{lemma}
\label{lemma:alphahatconv}
  If $\hat\alpha_n$ is of the form \eqref{eq:alphaest} and if the conditions (B) and (C') hold, then
   \[
    (nv_n)^{1/2}(\hat\alpha_n-\alpha_n)=\alpha\tilde Z_n(\phi_1)-\alpha^2\tilde Z_n(\psi)+o_P(1)
    \dto \alpha\tilde Z(\phi_1)-\alpha^2\tilde Z(\psi),
    \qquad n \to \infty,
   \]
   for a centered Gaussian process $\tilde Z$ with covariance function given by \eqref{eq:tildeZcov}.
\end{lemma}

Similarly as in \eqref{eq:bias1} and \eqref{eq:bias2}, we need an extra condition to ensure that the bias of $\hat\alpha_n$ is asymptotically negligible:
\begin{equation} \label{eq:bias3}
  \bigl\lvert
    \texpec{\tilde\psi(X_0/u_n)\mid X_0>u_n}- 1/\alpha
  \bigr\rvert
  = o \bigl( (nv_n)^{-1/2} \bigr),
  \qquad n \to \infty.
\end{equation}
Now we are ready to state the asymptotic normality of the backward estimator with estimated index $\alpha$, i.e.,
\[
  \CDFbAe{x}
  :=
  1 - \dfrac{\sum_{i=1}^{n}\left( \frac{X_{i-1}}{X_i}\right)^{\hat\alpha_n}\1\left(X_i/ X_{i-1} >x,
\; X_i > u_n \right)}{\sum_{i=1}^{n}\1\left(  X_i > u_n\right)}.
\]

\begin{theorem}
\label{cor:backwardestconv}
Suppose that the conditions of Corollary~\ref{cor:asnorm} and of Lemma~\ref{lemma:alphahatconv} are fulfilled and that~\eqref{eq:bias3} holds. Then
\begin{equation}
\label{eq:estconv2}
  (nv_n)^{1/2}
  \begin{pmatrix}
    (\CDFfA{x}- F^{(A_1)}(x))_{x\in[x_0,\infty)} \\
    (\CDFbAe{y}- F^{(A_1)}(y))_{y\in[y_0,\infty)}
  \end{pmatrix}
  \dto
  \begin{pmatrix}
    (Z^{(f,A_1)}(x))_{x\in[x_0,\infty)} \\
    (Z^{(\hat b,A_1)}(y))_{y\in[y_0,\infty)}
  \end{pmatrix},
  \qquad n \to \infty,
\end{equation}
with
\begin{align*}
  Z^{(f,A_1)}(x) &=  \tilde Z(\phi_{2,x})- \bar F^{(A_1)}(x)\tilde Z(\phi_1), \\
  Z^{(\hat b,A_1)}(y) &=  \tilde Z(\phi_{3,y})- \bar F^{(A_1)}(y)\tilde Z(\phi_1)+(\alpha^2\tilde Z(\psi)-\alpha\tilde Z(\phi_1))\expec{ \log (\Theta_1) \, \1(\Theta_1>y)}
\end{align*}
for a centered Gaussian process $\tilde Z$ with covariance function given by \eqref{eq:tildeZcov}.
\end{theorem}

The covariance function of the limiting process can be calculated in the same way as in the proof of Theorem~\ref{cor:asnorm}. In general, the resulting expressions will involve sums over all $k \in \NN$. Moreover, it is no longer guaranteed that the backward estimator of $ F^{(A_1)}(y)$ at $y>1$ has a smaller variance than  the forward estimator.

\begin{remark}
\label{rem:posneg}
For Markovian time series which are not necessarily positive, the forward and backward estimators of $ F^{(A_1)}$ and $ F^{(B_1)}$ can be represented in terms of generalized tail array sums constructed from
\[
  X_{n,i} := \frac{(X_{i-1},X_i,X_{i+1})}{u_n} \1(\abs{X_i} > u_n).
\]
When $x< 0$, for example, the backward estimator $\CDFbB{x}$ equals the ratio of the generalized tail array sums pertaining to
\begin{align*}
  \phi_{4,x}(y_{-1},y_0,y_1) &:= (-y_{-1}/y_0)^\alpha \, \1(y_0/y_{-1}\le x,y_0>1), \\
  \phi_5 (y_{-1},y_0,y_1) &:= \1(y_0<-1).
\end{align*}
Hence their limit processes can be obtained by the same methods as in the case $X_t>0$ under obvious analogues to the conditions (A($x_0$)), (B) and (C) with $v_n:=\prob{\abs{X_0}>u_n}$.
\end{remark}

\section{Examples}
\label{sec:examples}

We first show how, for  regularly varying Markov chains, the distributions of $A_{\pm 1}$ and $B_{\pm 1}$ can be calculated from the copula of $(X_0,X_1)$ if this copula satisfies certain regularity conditions. Next, we focus on solutions to stochastic recurrence equations, for which the Markov spectral tail chain exists and the distributions of $A_{\pm 1}$ and $B_{\pm 1}$, as well as $p$, admit a simple representation.


\subsection{Copula-based Markov processes}
\label{sec:examples:copula}

For stationary Markov processes, the joint distribution of $(X_t)_{t \in \ZZ}$ is determined by the law of $(X_0, X_1)$. In view of Sklar's theorem, the latter is determined by a univariate marginal distribution function, $G$, and a copula, $C$, via
\begin{equation}
\label{eq:CG}
  \prob{ X_0 \le x_0, \, X_1 \le x_1 }
  = C \bigl( G(x_0), \, G(x_1) \bigr), \qquad (x_0, x_1) \in \reals^2.
\end{equation}
Copula-based Markov models arise if the law of the Markov process is described by specifying $G$ and $C$ directly \citep{Chen2006,Chen2009}.

The result below links the distributions of $A_1$ and $B_1$ in \eqref{A1_prob} and \eqref{B1_prob} to the margin $G$ and the copula $C$. 
The result is framed in terms of the limit behaviour of $\dot{C}_1(u, v) = \partial C(u, v) / \partial u$ as $u$ tends to $0$ or $1$. The function $\dot{C}_1$ is related to the conditional distribution of $X_1$ given $X_0$; see~\eqref{eq:dotC1} below. For $z \ge 0$, consider the limits (whose existence is an assumption)
\begin{align}
\label{eq:CA1pos}
  \lim_{s \searrow 0} \dot{C}_1(1 - s, 1 - sz)
  &= \eta_{1,1}(z), \\
\label{eq:CA1neg}
  \lim_{s \searrow 0} \dot{C}_1(1 - s, sz)
  &= \eta_{1,0}(z),
\end{align}
and
\begin{align}
\label{eq:CB1pos}
  \lim_{s \searrow 0} \dot{C}_1(s, 1 - sz)
  &= \eta_{0,1}(z), \\
\label{eq:CB1neg}
  \lim_{s \searrow 0} \dot{C}_1(s, sz)
  &= \eta_{0,0}(z),
\end{align}
covering the four corners of the unit square; see Figure~\ref{fig:etas}.

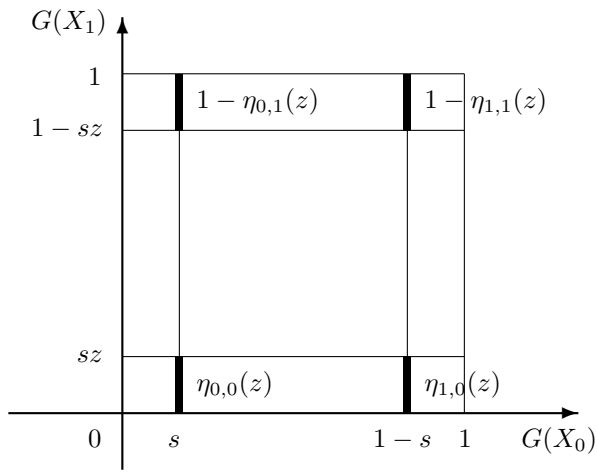
\begin{figure}
\setlength{\unitlength}{1.5cm}
\begin{center}
\begin{picture}(6,4.5)(-1,-0.5)
\thicklines
\put(0,-0.5){\vector(0,1){4}}
\put(-1,0){\vector(1,0){5}}
\thinlines
\put(3,0){\line(0,1){3}}
\put(0,3){\line(1,0){3}}
\put(0.5,0){\line(0,1){3}}
\put(0,0.5){\line(1,0){3}}
\put(2.5,0){\line(0,1){3}}
\put(0,2.5){\line(1,0){3}}
\linethickness{1mm}
\put(0.5,0){\line(0,1){0.5}}
\put(0.5,2.5){\line(0,1){0.5}}
\put(2.5,0){\line(0,1){0.5}}
\put(2.5,2.5){\line(0,1){0.5}}
\put(-0.3,-0.3){$0$}
\put(-0.3,2.9){$1$}
\put(2.95,-0.3){$1$}
\put(0.4,-0.3){$s$}
\put(2.2,-0.3){$1-s$}
\put(-0.4,0.45){$sz$}
\put(-0.8,2.45){$1-sz$}
\put(2.65,0.2){$\eta_{1,0}(z)$}
\put(2.65,2.7){$1-\eta_{1,1}(z)$}
\put(0.65,0.2){$\eta_{0,0}(z)$}
\put(0.65,2.7){$1-\eta_{0,1}(z)$}
\put(3.5,-0.3){$G(X_0)$}
\put(-0.8,3.4){$G(X_1)$}
\end{picture}
\end{center}
\caption{\label{fig:etas} \footnotesize
Viewing the functions $\eta_{0,0}$, $\eta_{0,1}$, $\eta_{1,0}$ and $\eta_{1,1}$ in terms of the conditional distribution of $G(X_1)$ conditionally on $G(X_0) = s$ or on $G(X_0) = 1-s$.}
\end{figure}

\begin{proposition}
\label{prop:copMarkov}
Let the distribution of $(X_0, X_1)$ be given by \eqref{eq:CG}, where $X_0$ has a Lebesgue density and satisfies the regular variation condition \eqref{eq:RV} and the tail-balance condition~\eqref{p_prob}. Assume that $C$ admits a continuous first-order partial derivative $\dot{C}_1$ on $(0, 1) \times [0, 1]$. If $p > 0$ and if the limits in \eqref{eq:CA1pos} and \eqref{eq:CA1neg} exist and are continuous on $ [0, \infty)$, then \eqref{A1_prob} holds and
\begin{equation}
\label{eq:A1:eta}
  \prob{ A_1 \le x } =
  \begin{cases}
    \eta_{1,1}(x^{-\alpha}) & \text{if $x > 0$,} \\[1ex]
    \eta_{1,0}( \frac{1-p}{p} \, \abs{x}^{-\alpha}) & \text{if $x < 0$.}
  \end{cases}
\end{equation}
Similarly, if $p < 1$ and if the limits in \eqref{eq:CB1pos} and \eqref{eq:CB1neg} exist and are continuous on $[0, \infty)$, then \eqref{B1_prob} holds and
\begin{equation}
\label{eq:B1:eta}
  \prob{ B_1 \le x } =
  \begin{cases}
    1-\eta_{0,0}(x^{-\alpha}) & \text{if $x > 0$,} \\[1ex]
    1-\eta_{0,1}( \frac{p}{1-p} \, \abs{x}^{-\alpha}) & \text{if $x < 0$.}
  \end{cases}
\end{equation}
\end{proposition}

Similarly, the distributions of $A_{-1}$ and $B_{-1}$ in \eqref{A-1_prob} and \eqref{B-1_prob} can be obtained via the limit behaviour of $\dot{C}_2(u, v) = \partial C(u, v) / \partial v$ as $v$ tends to $0$ or $1$. Some examples are worked out in Appendix~\ref{sec:additional}, whereas the proof of Proposition~\ref{prop:copMarkov} is given in Appendix~\ref{sec:proofs}.

\subsection{Stochastic recurrence equations}
\label{sec:examples:sre}

The stochastic recurrence equation
\begin{equation}\label{eq_SRE}
X_t=C_t X_{t-1} + D_t, \quad t\in\mathbb{Z},
\end{equation}
received some attention in time series analysis and extreme value theory. We focus on the case where $\left(C_t,D_t\right)$, $t\in\mathbb{Z}$, is an i.i.d.\ $\mathbb{R}^2$--valued sequence. Provided that $-\infty \leq \expec{\log\vert C_1\vert} <0$ and $\expec{\log^+ \vert D_1\vert} <\infty$, where $\log^+ x=\text{max}(\log x,0 )$, there exists a unique strictly stationary causal solution to \eqref{eq_SRE} \citep[Corollary~2.2]{basrak:davis:mikosch:2002}.

Results on regular variation of $X_0$ were first developed by \cite{Kesten73}. His Theorem~5 states that if there exists $\alpha>0$ such that $\expec{\vert C_1 \vert^\alpha} =1$, $\expec{\vert C_1\vert^\alpha \log^+ \vert C_1 \vert} < \infty$ and $\expec{\vert D_1 \vert^\alpha} < \infty$ and if some other conditions are satisfied, then $X_0$ is regularly varying; more specifically,
\begin{equation}\label{eq:rv_kesten}
  \left.
  \begin{array}{rcl}
    \Pr(X_0 > x) &=& c_+ x^{-\alpha} \, (1 + o(1)) \\[1ex]
    \Pr(X_0 \leq -x) &=& c_- x^{-\alpha} \, (1 + o(1))
  \end{array}
  \right\},
  \qquad x \to \infty,
\end{equation}
for constants $c_+$, $c_- \geq 0$ such that $c_+ + c_- >0$.
This result was extended by \citet[Lemma~2.2 and Theorem~4.1]{goldie1991}, who gave explicit expressions for $c_+$ and $c_-$. Regular variation of $X_0$ and iteration of \eqref{eq_SRE} gives joint regular variation of $(X_t)_{t \in \ZZ}$ with index $\alpha$.

The forward spectral tail process $(\Theta_t)_{t \ge 0}$ of $(X_t)_{t \in \ZZ}$ admits the representation
\begin{equation*}
  \Theta_t = \Theta_0 \prod_{h=1}^t \tilde{C}_h, \qquad t\in\NN,
\end{equation*}
where $\tilde{C}_t$, $t\in\NN$, are i.i.d.\ random variables with the same distribution as $C_1$ which are independent of $\Theta_0$ \citep[Example~6.1]{segersjanssen}. Hence, $(\Theta_t)_{t \geq 0}$ becomes a Markov spectral tail chain satisfying \eqref{eq:Markov:forward} with $\law{A_1} = \law{B_1} = \law{C_1}$. Moreover, by \citet[Theorem~4.1]{goldie1991},  if $\prob{ C_1 < 0 } > 0$ then necessarily $c_+ = c_-$ and therefore $p = \prob{ \Theta_0 = +1 } = 1/2$. This property also follows from the following result on more general tail spectral processes; its proof is given in Appendix~\ref{sec:proofs}.

\begin{lemma}
\label{lem:tailsign}
Let $(\Theta_t)_{t \in \ZZ}$ be the spectral tail process of an $\alpha$-regularly varying stationary time series $(X_t)_{t \in \ZZ}$. Suppose the following two conditions hold:
\begin{enumerate}[(i)]
\item
$\expec{\abs{\Theta_1}^\alpha} = 1$;
\item
$\Theta_1 / \Theta_0$ and $\Theta_0$ are independent.
\end{enumerate}
Then $\Theta_{-1} / \Theta_0$ and $\Theta_0$ are independent too.  Moreover, $\prob{ \Theta_1/\Theta_0 < 0 } > 0$ implies $p = \prob{ \Theta_0 = 1 } = 1/2$.
\end{lemma}



Lemma~\ref{lem:tailsign} applies to the solution $(X_t)_{t \in \ZZ}$ of the stochastic recurrence equation \eqref{eq_SRE}. Above, we had already seen that $\law{A_1} = \law{B_1} = \law{C_1}$ and $\expec{\abs{C_1}^\alpha} = 1$. Hence, the spectral tail process $(\Theta_t)_{t \in \ZZ}$ satisfies the two conditions in Lemma~\ref{lem:tailsign}. We obtain that $\Theta_{-1}/\Theta_0$ is independent of $\Theta_0$ too. But since the backward spectral tail process admits the representation in \eqref{eq:Markov:backward}, we conclude that $\law{A_{-1}} = \law{B_{-1}}$. If $p = 0$ or $p = 1$, the previous statements simplify, since $A_1$ and $A_{-1}$ in \eqref{A1_prob}--\eqref{A-1_prob} are defined only if $p > 0$ whereas $B_1$ and $B_{-1}$ in \eqref{B1_prob}--\eqref{B-1_prob} are defined only if $p < 1$. For $p \in \{0, 1\}$, the law of $\Theta_{0}$ is degenerate at $-1$ or at $1$, respectively, so that the statements on independence are trivially satisfied.

\section{Monte Carlo simulations}
\label{sec:simul}

To compare the finite-sample performance of the forward, backward, and mixture estimators, we conducted a Monte Carlo simulation study. We use both approaches discussed at the end of Section~\ref{sec:estim} for dealing with the problem that $\alpha$ is unknown.

Pseudo-random samples are generated from two different time series models.
\begin{itemize}
\item
For the copula Markov model (Section~\ref{sec:examples:copula}), we choose the symmetric t-distribution with $\nu_1=2$ degrees of freedom as the margin $G$ and the t-copula $C^t_{\nu_2,\rho}(u,v)$ with $\nu_2=2.5$ and $\rho=0.2$ to model the temporal dependence structure. Hence, $\prob{\Theta_0=1}=\prob{\Theta_0=-1}=1/2$ and the index of regular variation is equal to $2$. The distribution function of $A_1$ is calculated in Example~\ref{ex:tcopula}. The simulation algorithm is described in \citet[Section~6.2]{Chen2009}.
\item
For the stochastic recurrence equation (Section~\ref{sec:examples:sre}), we let $C_t$ and $D_t$ be independent $N(1/3, 8/9)$ and $N(-10, 1)$ random variables, respectively. This choice ensures that $\expec{C_t^2}=1$ and that the sufficient conditions of Theorem~5 in \cite{Kesten73} hold. As a consequence, a stationary solution $\left(X_t\right)_{t\in \ZZ}$ exists which is regularly varying with index $\alpha=2$.
\end{itemize}

For each model, we generate time series $(X_i)_{1\le i\le n}$ of length $n = 2\,000$ and set the threshold $u_n$ to be the $97.5\%$ quantile of the absolute values of the sample, i.e., we use 50 extremes for estimation. Based on $1\,000$ Monte Carlo repetitions, we estimate the bias, the standard deviation ($\SD$) and the root mean squared error ($\RMSE$) of all estimators under consideration. Note that both models are symmetric in the sense that $\law{A_1} = \law{B_1}$. Moreover, the estimators of the cdf's of $A_1$ and $B_1$ are identically distributed for the copula model, and they behave similarly for the solutions to the stochastic recurrence equation. Therefore, we report the results only for $A_1$.

\begin{figure}[htp]
\begin{center}
\begin{tabular}{ccc}
\includegraphics[width=0.3\textwidth]{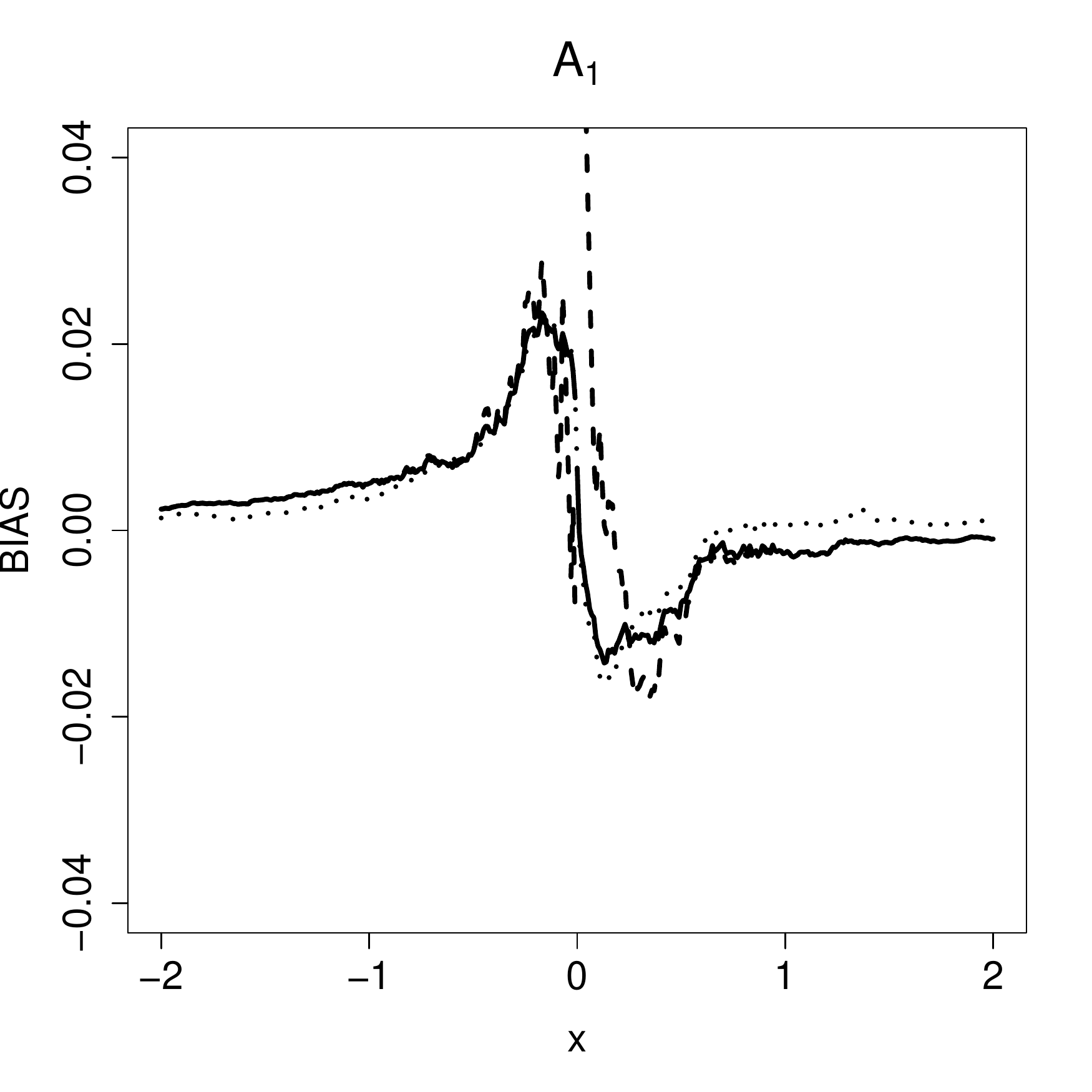}
&
\includegraphics[width=0.3\textwidth]{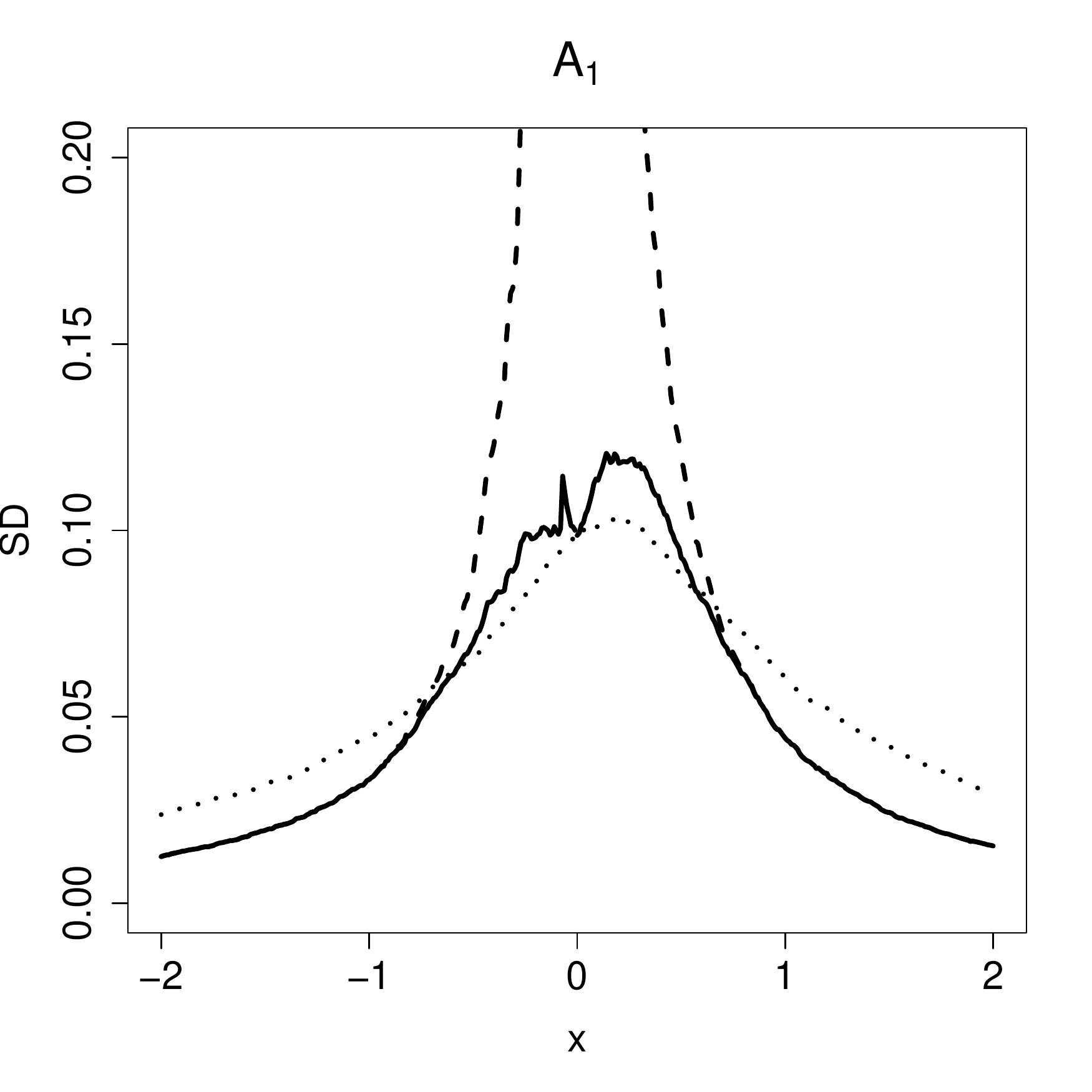}
&
\includegraphics[width=0.3\textwidth]{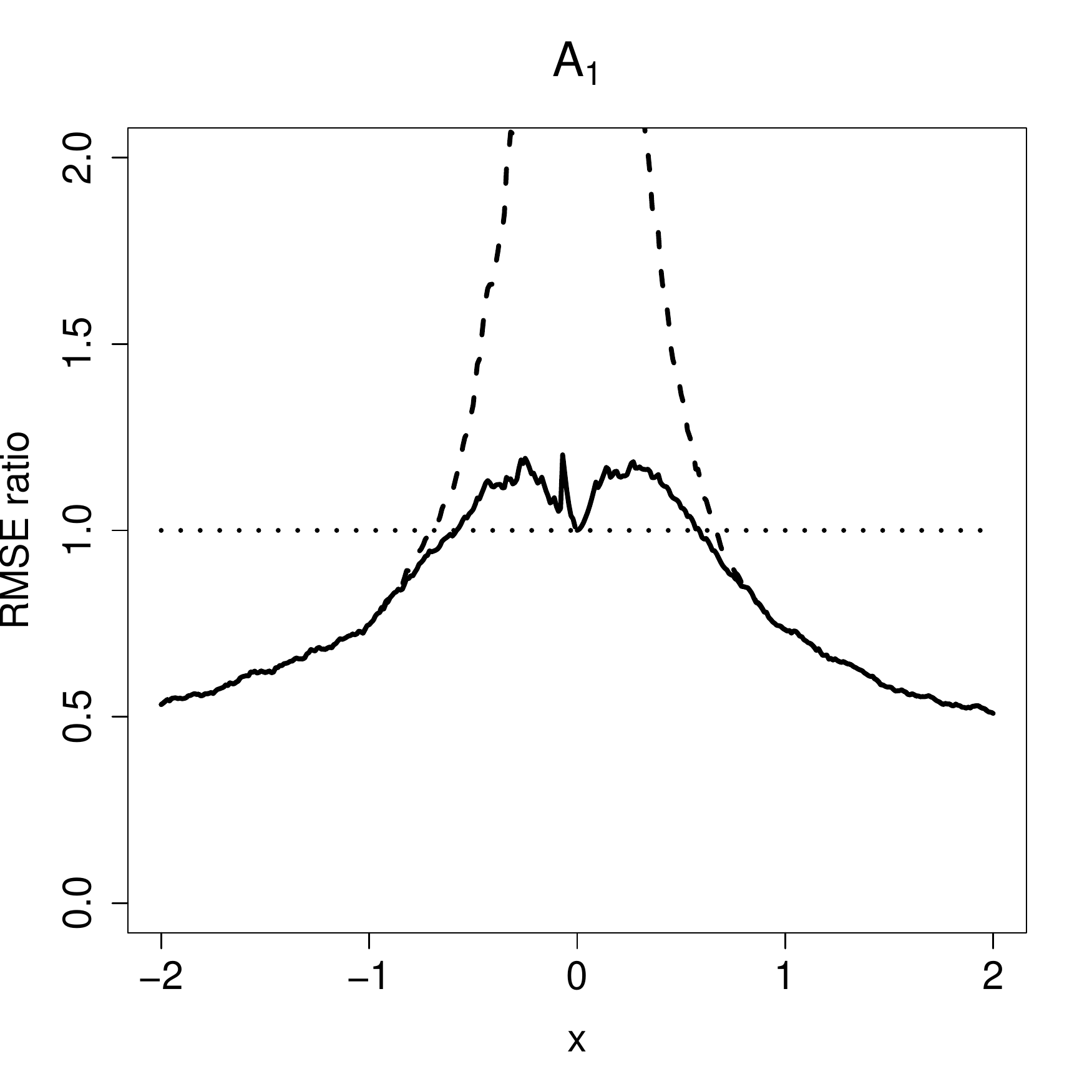}
\\
\includegraphics[width=0.3\textwidth]{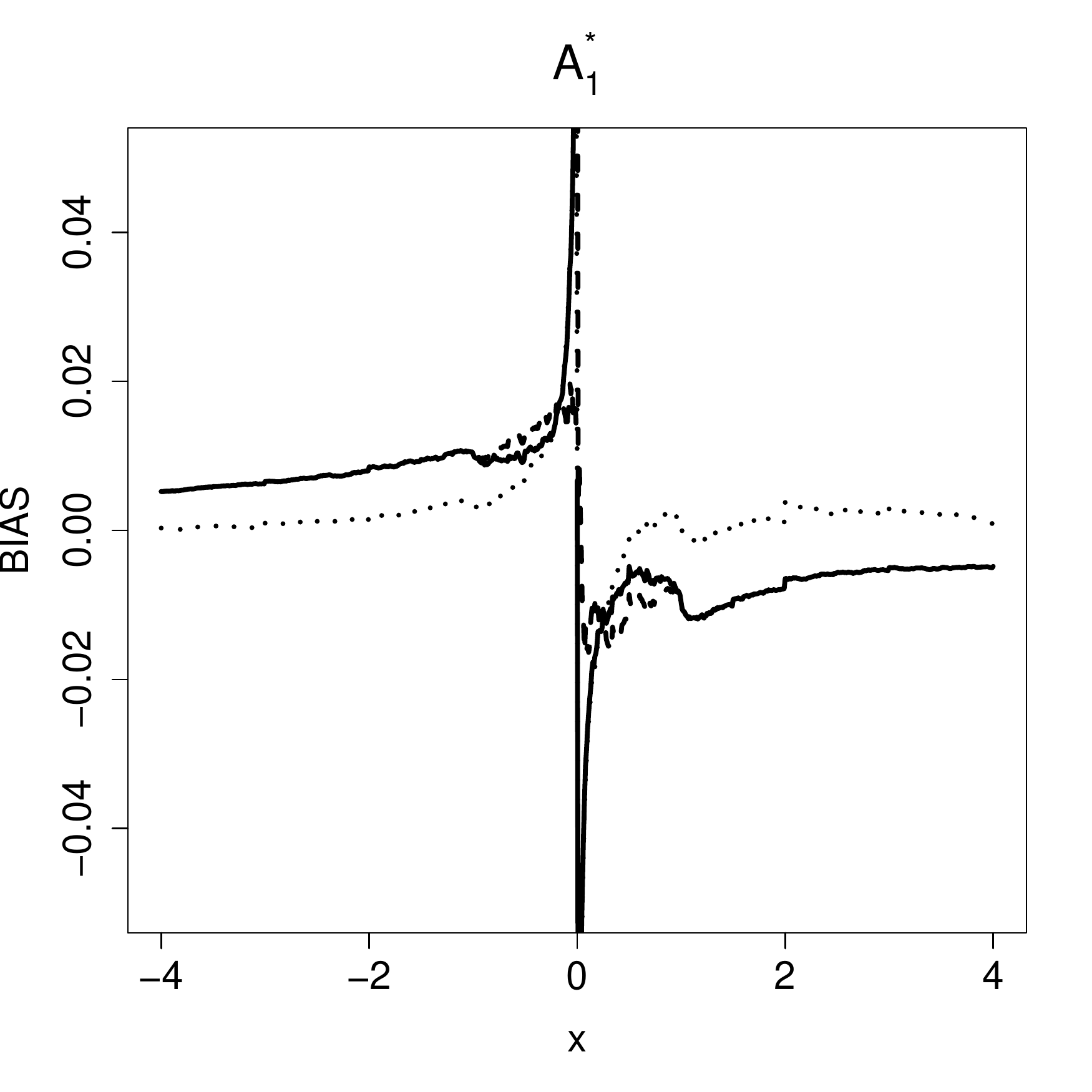}
&
\includegraphics[width=0.3\textwidth]{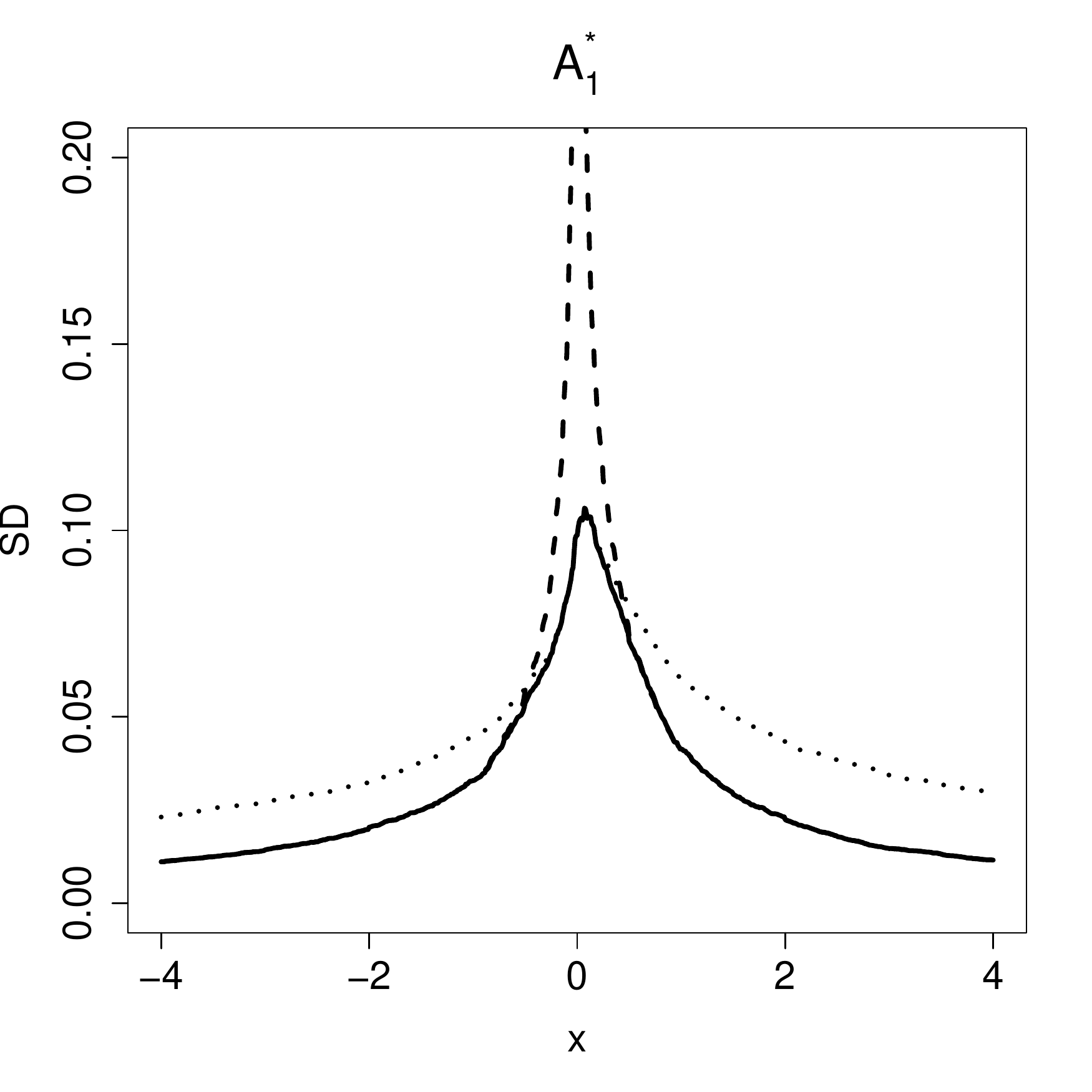}
&
\includegraphics[width=0.3\textwidth]{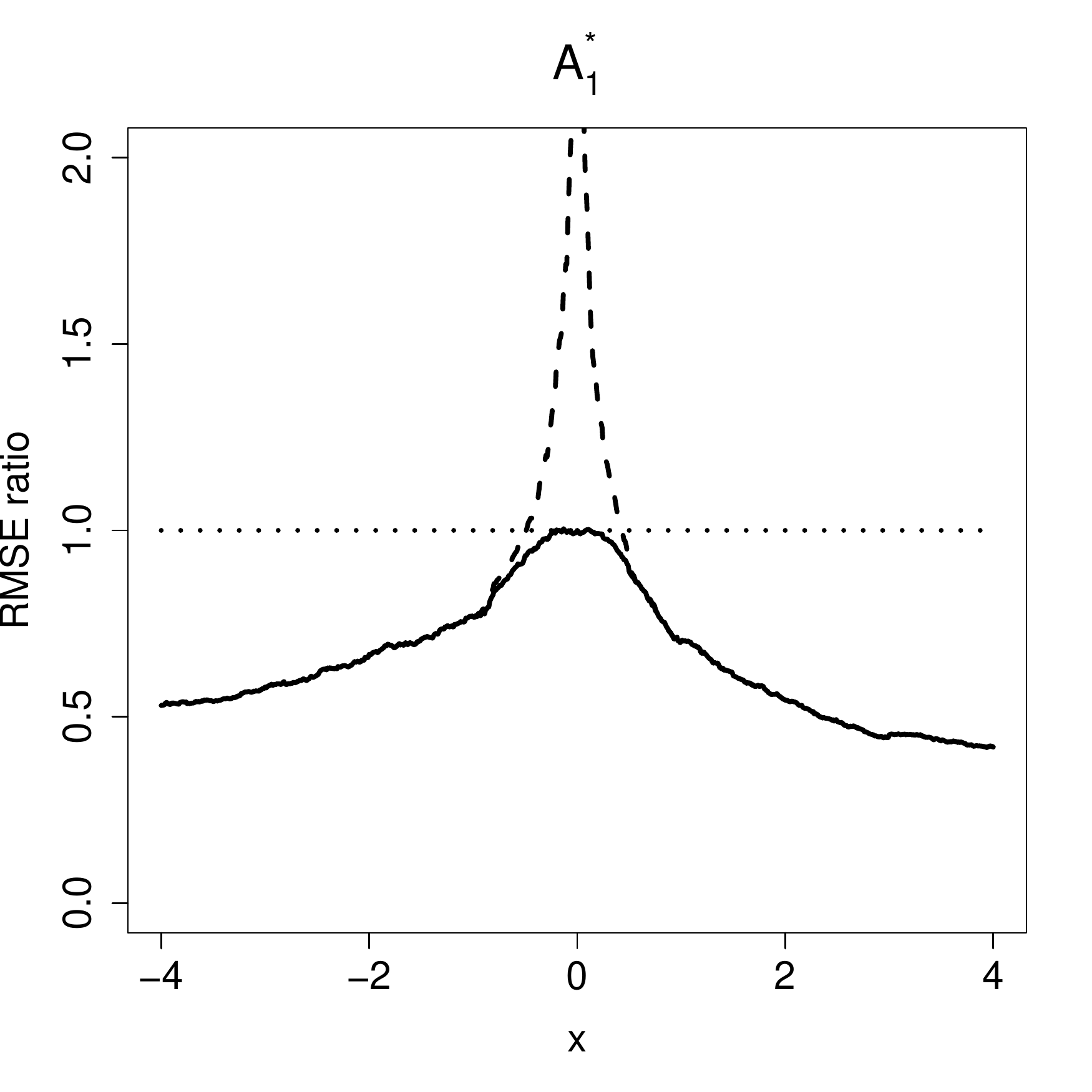}
\end{tabular}
\end{center}
\caption{\label{simu:CMM:a} \footnotesize
Simulation results for the copula Markov model. Top: Estimation of the cdf of $A_1$ using the plug-in estimator for $\alpha$. Bottom: Estimation of the cdf of $A_1^* = \sign(A_1) \abs{A_1}^\alpha$ with the signed rank transformation. Left: bias; middle: standard deviation; right: root mean squared error ratio with respect to the forward estimator. Solid line: mixture estimator; dashed line: backward estimator; dotted line: forward estimator.}
\end{figure}

\begin{figure}[htp]
\begin{center}
\begin{tabular}{ccc}
\includegraphics[width=0.3\textwidth]{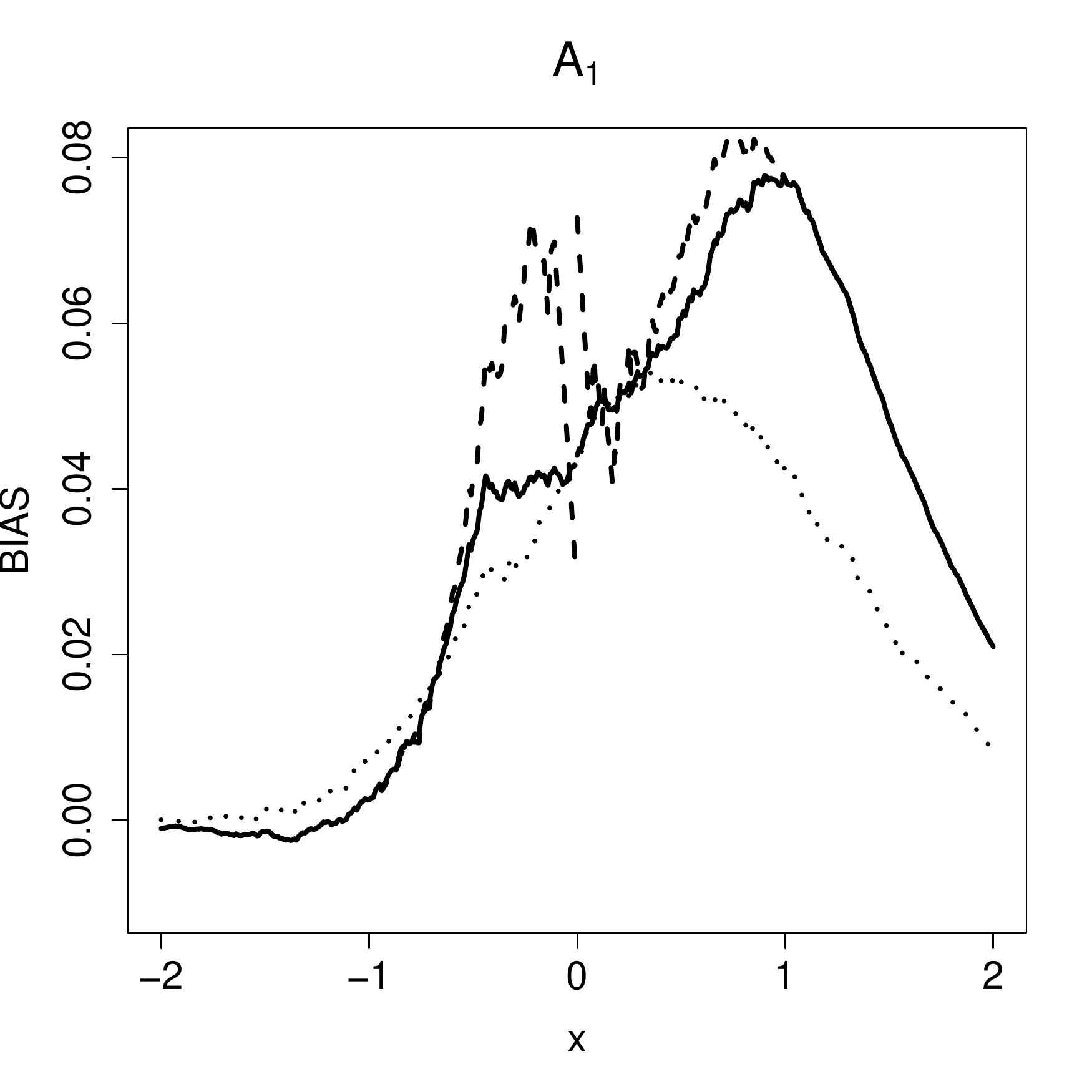}
&
\includegraphics[width=0.3\textwidth]{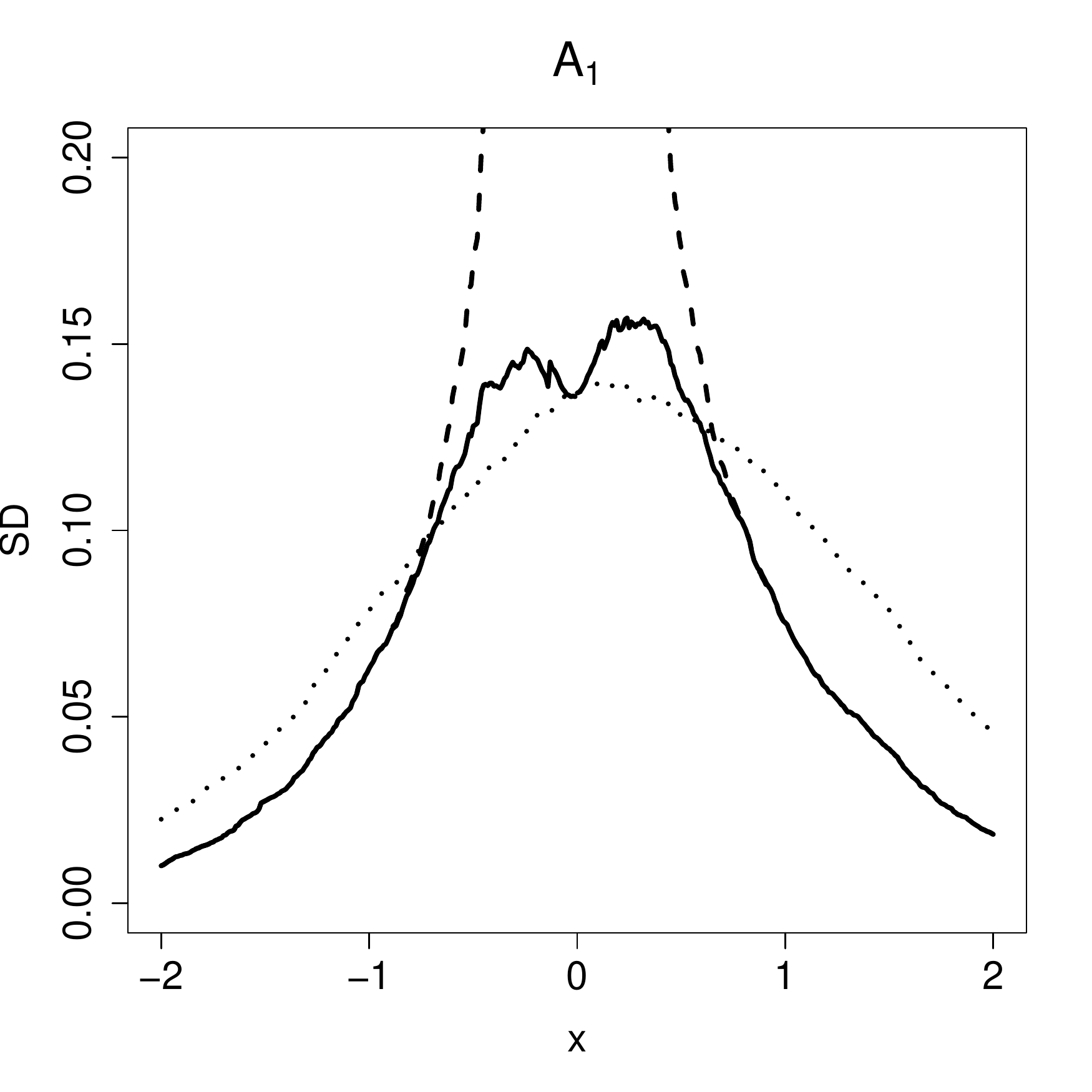}
&
\includegraphics[width=0.3\textwidth]{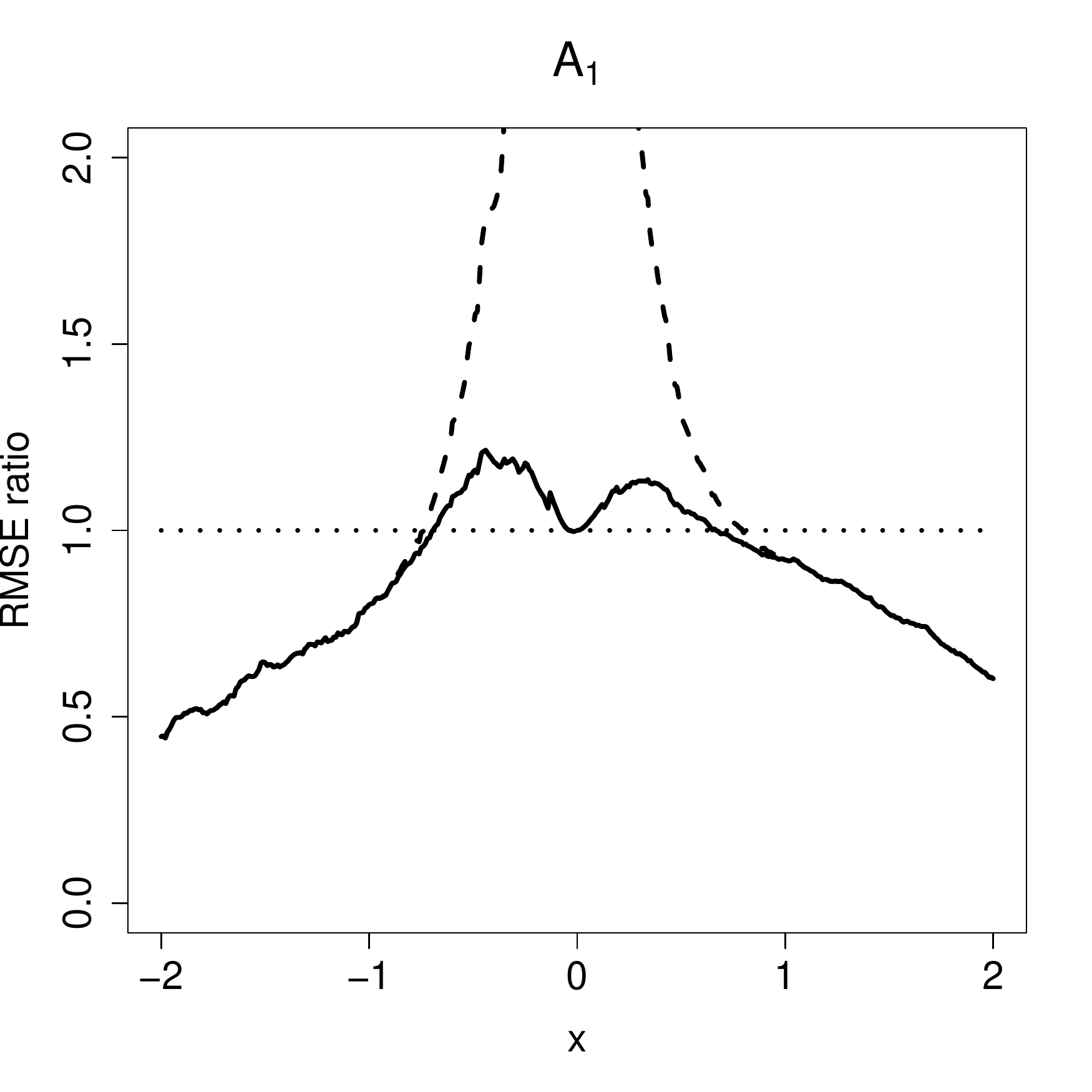}
\\
\includegraphics[width=0.3\textwidth]{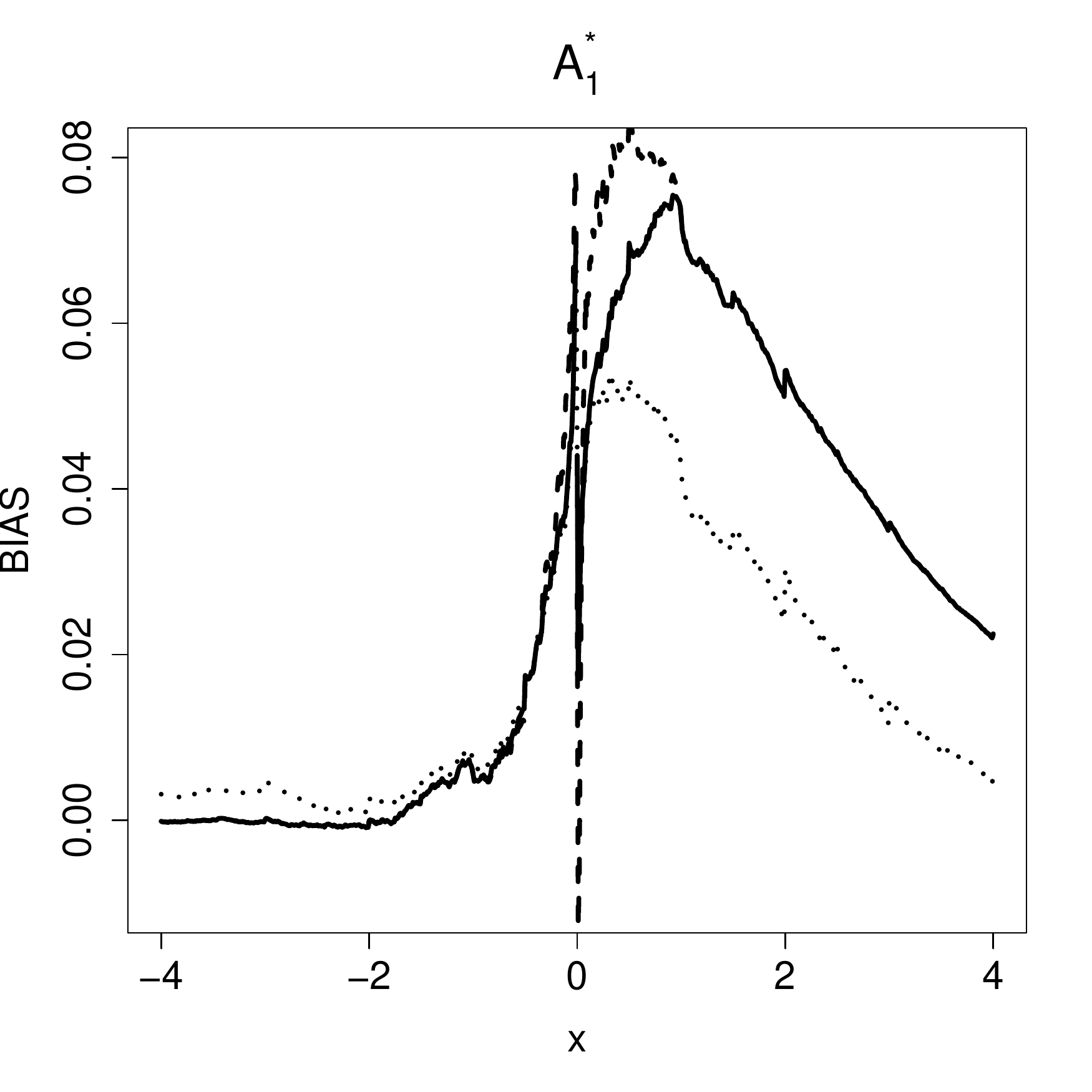}
&
\includegraphics[width=0.3\textwidth]{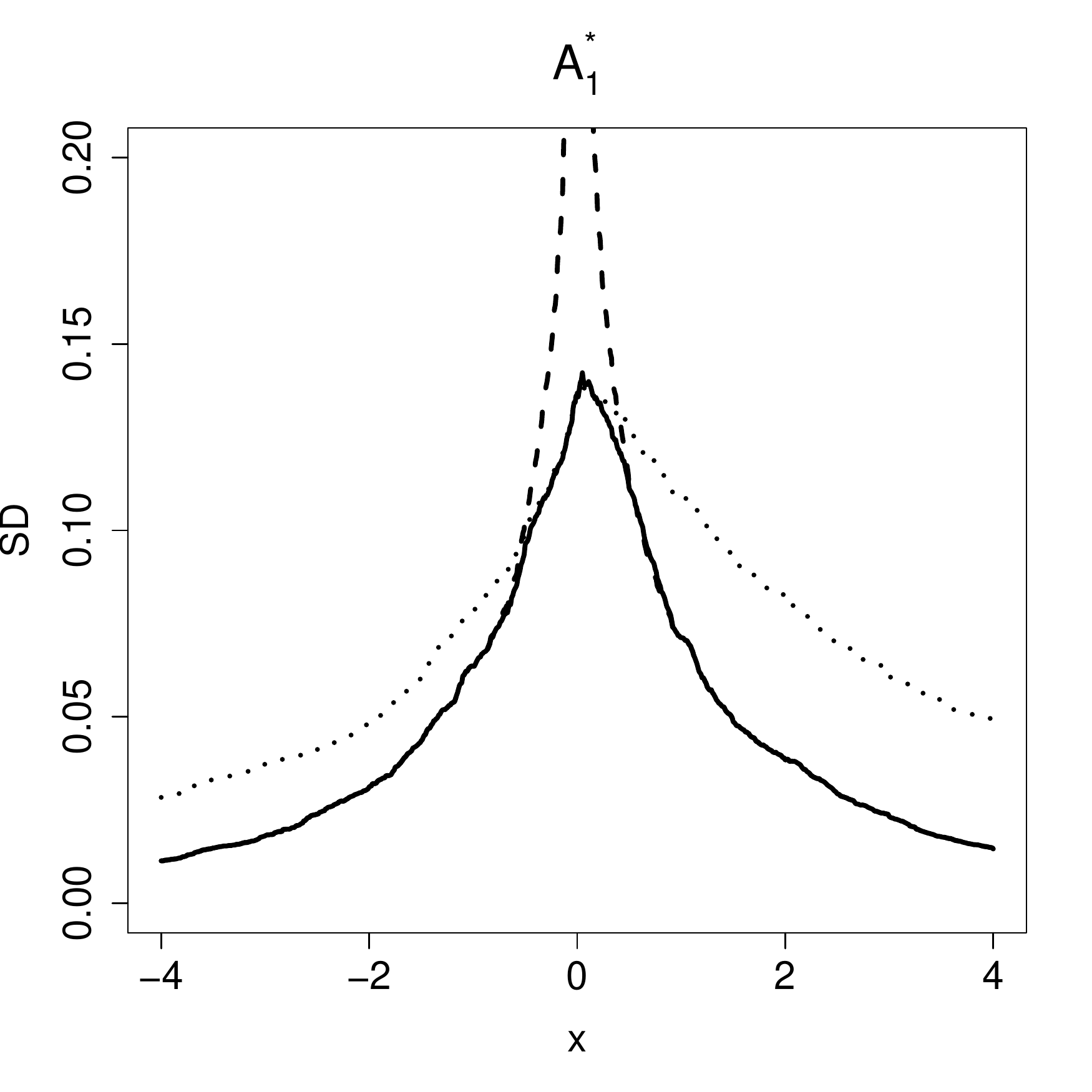}
&
\includegraphics[width=0.3\textwidth]{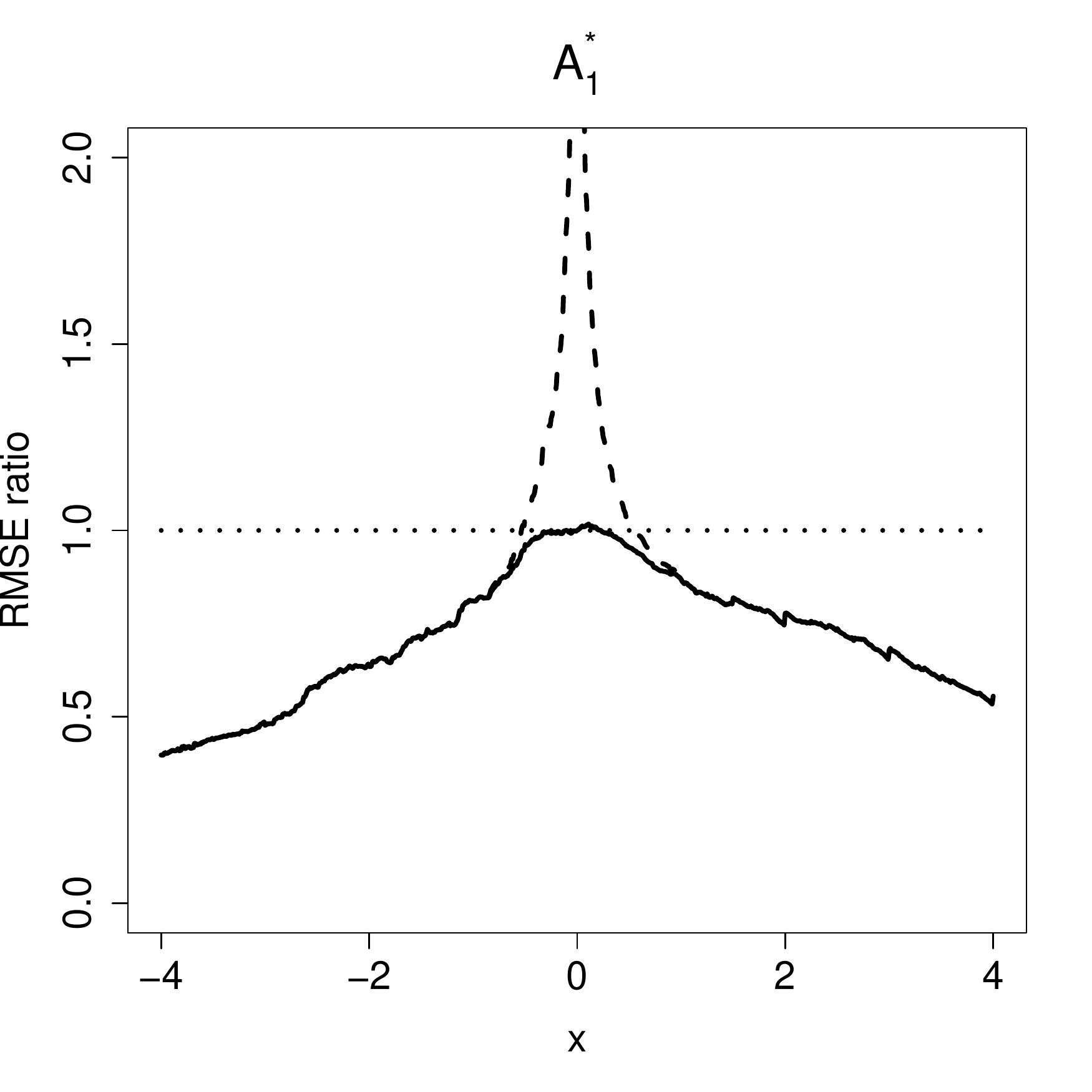}
\end{tabular}
\end{center}
\caption{\label{simu:SRE:a} \footnotesize
Simulation results for the stochastic recurrence equations. Top: Estimation of the cdf of $A_1$ using the plug-in estimator for $\alpha$. Bottom: Estimation of the cdf of $A_1^* = \sign(A_1) \abs{A_1}^\alpha$ with the signed rank transformation. Left: bias; middle: standard deviation; right: root mean squared error ratio with respect to the forward estimator. Solid line: mixture estimator; dashed line: backward estimator; dotted line: forward estimator.}
\end{figure}

Figure~\ref{simu:CMM:a} shows the results for the estimators of the cdf's of $A_1$ and $A_1^*$ in the copula Markov model. For this model, $\Bias(\hat{p})=0.001$, $\SD(\hat{p})=0.08$, and $\RMSE(\hat{p})=0.08$. In addition, for the simulations with estimated $\alpha$, we obtain $\Bias(\hat{\alpha})=0.077$, $\SD(\hat{\alpha})=0.407$, and $\RMSE(\hat{\alpha})=0.414$. In the top row, the Hill-type estimator \eqref{Hill} is used, whereas in the bottom row, we apply the signed rank transformation \eqref{emp_trans}. The plots on the left and in the middle show the bias and the standard deviation of the forward estimator as a dotted line, of the backward estimator as a dashed line and of the mixture estimator as a solid line. The plots on the right-hand side display the relative RMSE of the backward and the mixture estimator w.r.t.\ the forward estimator.

As expected, the backward estimator outperforms the forward estimator for $\abs{x}\ge 1$, whereas for arguments $x$ close to $0$ the performance of the backward estimator is poor and it is the forward estimator that has the lower RMSE. The mixture estimator performs much better than the backward estimator for small $\abs{x}$, while for large values of $\abs{x}$ its RMSE is similar to the one of the backward estimator. Indeed, in the approach using rank transforms, the RMSE of the mixture estimator is never much larger than that of the forward estimator, while the mixture estimator is almost twice as efficient as the forward estimator for $\abs{x} > 2$. Hence, it clearly outperforms both the forward and the backward estimator in this approach. In contrast, the mixture estimator performs worse than the forward estimator for small values of $\abs{x}$ if $\alpha$ is replaced with the Hill-type estimator.

Figure~\ref{simu:SRE:a} shows the analogous results for the solutions of the stochastic recurrence equation. We obtain $\Bias(\hat{p})=-0.209$, $\SD(\hat{p})=0.083$, $\RMSE(\hat{p})=0.225$, and $\Bias(\hat{\alpha})=0.173$, $\SD(\hat{\alpha})=0.526$, $\RMSE(\hat{\alpha})=0.553$. By and large, the relative performance of all estimators is similar to the one in the copula Markov model. The relative efficiency of the mixture estimator with respect to the forward estimator is a bit higher for small values of $\abs{x}$, while it is slightly worse for larger values. Overall, the absolute estimation errors of all estimators are about $20$ to $30\%$ larger.

\section{Case study}
\label{sec:applic}

The spectral tail process of a heavy tailed time series conveys important information on its serial extremal dependence. Such extremal dependence typically arises e.g.\ in financial time series which exhibit clustering of extremes. By estimating the joint distribution of $(\Theta_0,\Theta_1)$ and $(\Theta_0,\Theta_{-1})$, we gain insight into short-range dependence between extremes of consecutive observations, covering both lower and upper tails.

An interesting question is whether the distribution of a time series remains unaffected by reversing the direction of time. It is well understood how to test for such time reversibility regarding the bulk of the distribution \citep{beare2012,Chen_Ray_2000}. Clearly, if a time series does not look time reversible, it should not be modelled with processes having this property, such as, e.g., stationary Gaussian processes. Financial time series do not have this feature in general \citep{Chen_Kuan_2002}. This, however, does not imply that there can be no time reversibility at extreme levels. Formally, such extremal time reversibility would mean that $\law{\Theta_0,\Theta_{1},\ldots,\Theta_{t}} = \law{\Theta_0,\Theta_{-1},\ldots,\Theta_{-t}}$. In particular, it implies that $\law{\Theta_0,\Theta_{1}} = \law{\Theta_0,\Theta_{-1}}$ and, equivalently, $\law{A_1} = \law{A_{-1}}$ and $\law{B_1} = \law{B_{-1}}$, or $\law{A_1^*} = \law{A_{-1}^*}$ and $\law{B_1^*} = \law{B_{-1}^*}$. Comparison of the estimated distribution functions may be used to reject the extremal time reversibility hypothesis, and thus can serve as a diagnostic tool for model selection. 

Note that under Markov assumption $\law{\Theta_0,\Theta_{1},\ldots,\Theta_{t}} = \law{\Theta_0,\Theta_{-1},\ldots,\Theta_{-t}}$ is equivalent to $\law{\Theta_0,\Theta_{1}} = \law{\Theta_0,\Theta_{-1}}$, and hence inference about extremal time reversibility can be based only on the pairwise comparison, i.e., $\law{A_1^*} = \law{A_{-1}^*}$ and $\law{B_1^*} = \law{B_{-1}^*}$. However, Markovianity is rarely a realistic assumption  for financial time series.

We analyze daily log-returns of Google and UBS stock prices between 2005-01-03 and 2013-12-31 (taken from \url{www.google.com/finance}), leading to  $2279$ observations for Google and $2280$ for UBS. The thresholds are set at the $95\%$ quantiles, giving 114 extremes. For Google log-returns we obtain 53 positive extremes and 61 negative extremes, whereas for UBS it is 50 and 64 positive and negative extremes respectively. The estimated index of regular variation is equal to $2.88$ for Google and $2.51$ for UBS.

We discuss the results jointly for both stocks, as their extremal dependence structures exhibit similar patterns. The two data series are presented in Figure~\ref{app1}. The three rows show the daily closing prices (top), the daily log-returns $X_i$ (middle) and the pertaining rank transformed log-returns $X_{n,i}^*$ as in \eqref{emp_trans} (bottom).

\begin{figure}[htp]
\begin{center}
\includegraphics[scale=0.8]{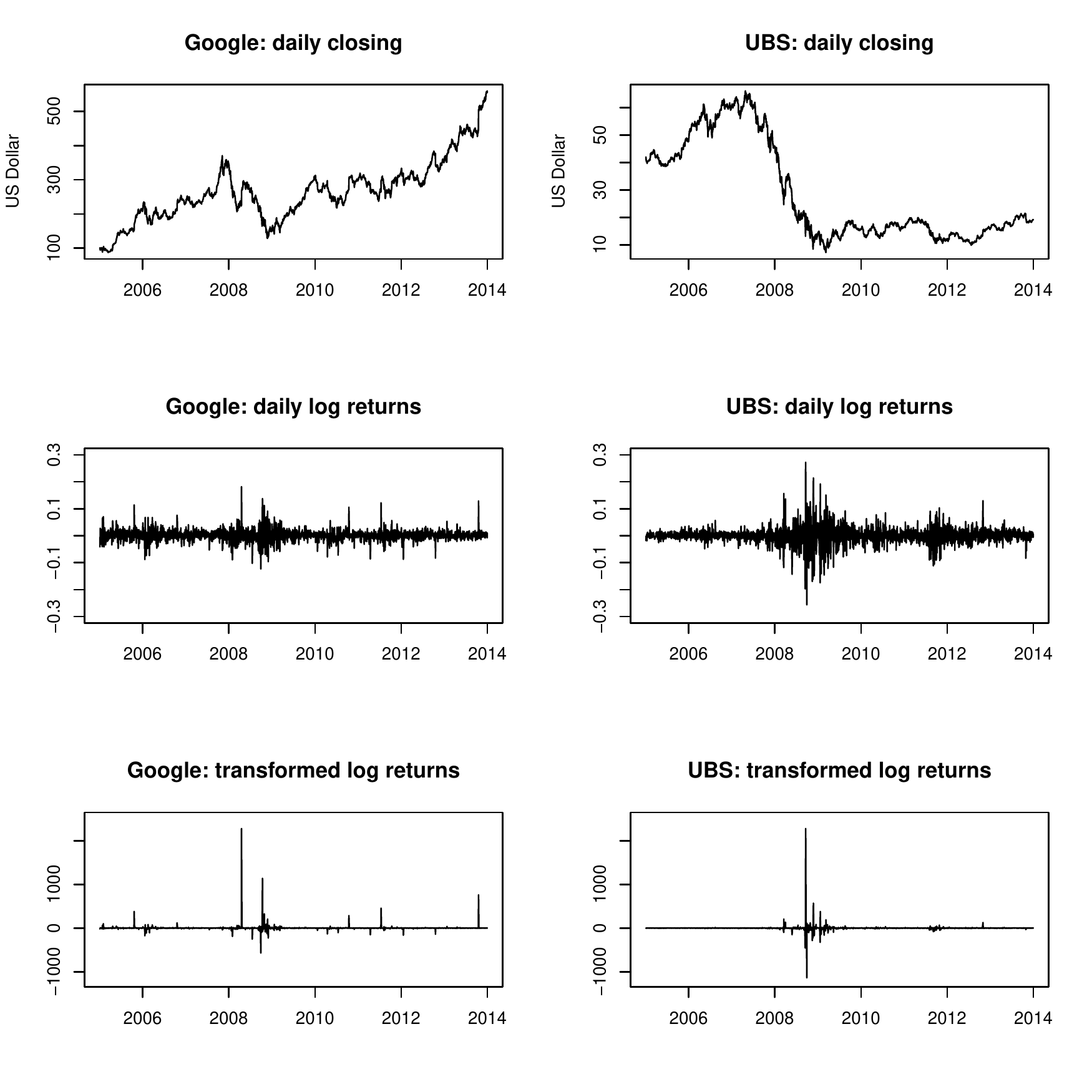}
\caption{\footnotesize{Google and UBS daily stock prices, daily log-returns and transformed log-returns.}\label{app1}}
\end{center}
\end{figure}
The distribution functions of $A_{-1}^*$ and $B_{-1}^*$ were estimated using the reversed time series. For $A_{1}^*$, $A_{-1}^*$, $B_{1}^*$, and $B_{-1}^*$, we applied a monotonized version of the mixture estimator \eqref{mixture_A1}, which otherwise is not necessarily monotone. The monotonized version is defined by the smallest increasing function larger than or equal to the mixture estimator, for $x\geq 0$, and the biggest increasing function smaller than or equal to the mixture estimator, for $x<0$. We display the estimated distribution functions of $A_{1}^*$, $A_{-1}^*$, $B_{1}^*$, and $B_{-1}^*$ for the rank transformed log-returns based on the monotonized mixture estimator in Figure~\ref{app6}.

Since the cdf's of $A_1^*$ and $A_{-1}^*$ as well as the cdf's of $B_1^*$ and $B_{-1}^*$ apparently differ on the negative real line, the extreme values of the log-returns exhibit no time reversibility. Moreover, most of the probability mass of $A_{1}^*$ and $B_{-1}^*$ is concentrated near the origin. This hints at asymptotic independence of $X_t^+$ and $X_{t+1}$, and of $X_t$ and $X^-_{t+1}$. In contrast, the laws of $A_{-1}^*$ and $B_1^*$ put considerable mass on the negative half-line. This indicates that with asymptotically non-negligible probability a large loss is succeeded by a large gain. Such an event can be interpreted as a correction to an overreaction by the stock market. Indeed, the existence of long--term and short--term overreaction behaviour in various stock markets is well documented \citep{bondt1985, Bowman1998, ali2010}.

\begin{figure}[htp]
\begin{center}
\begin{tabular}{cc}
\includegraphics[width=0.4\textwidth]{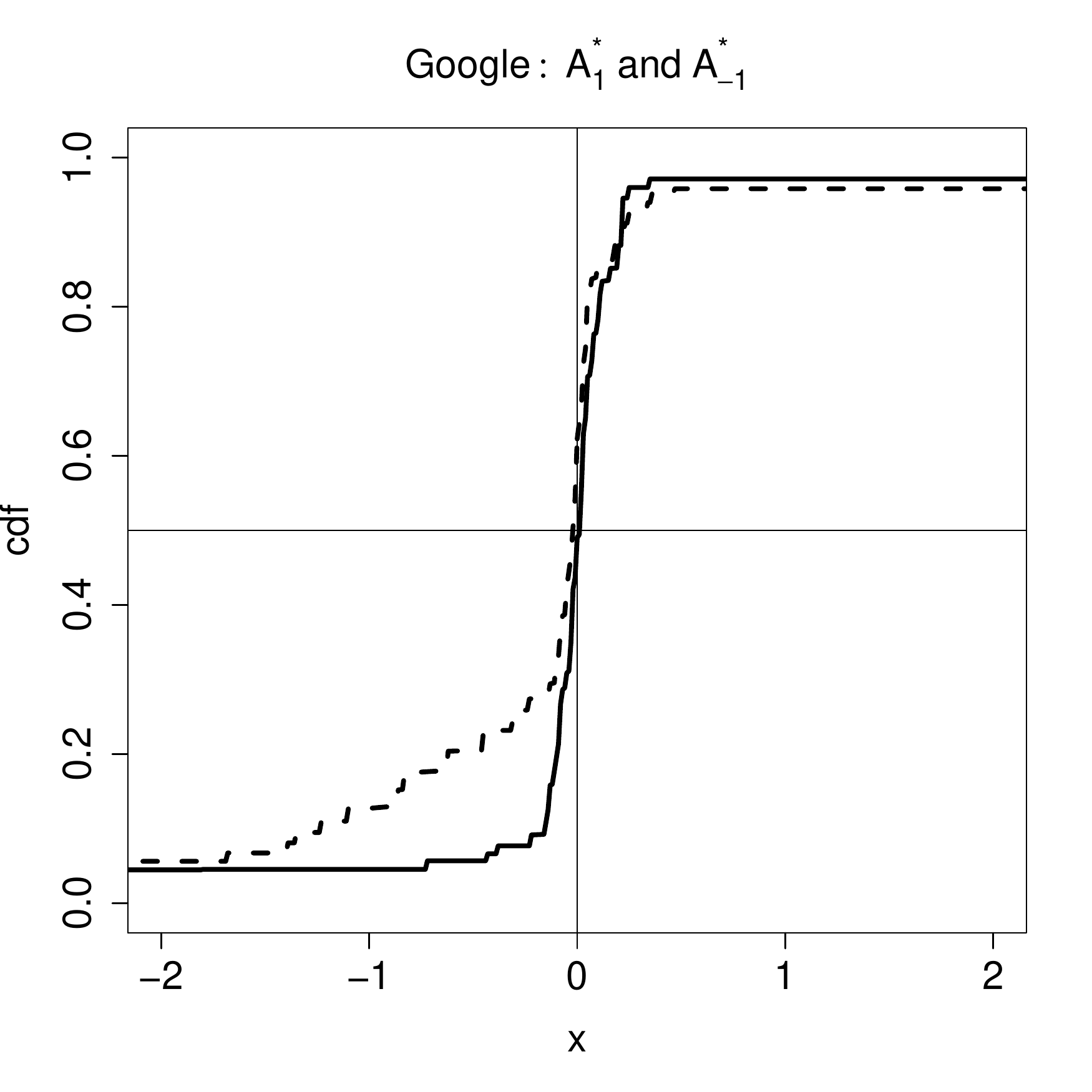}
&
\includegraphics[width=0.4\textwidth]{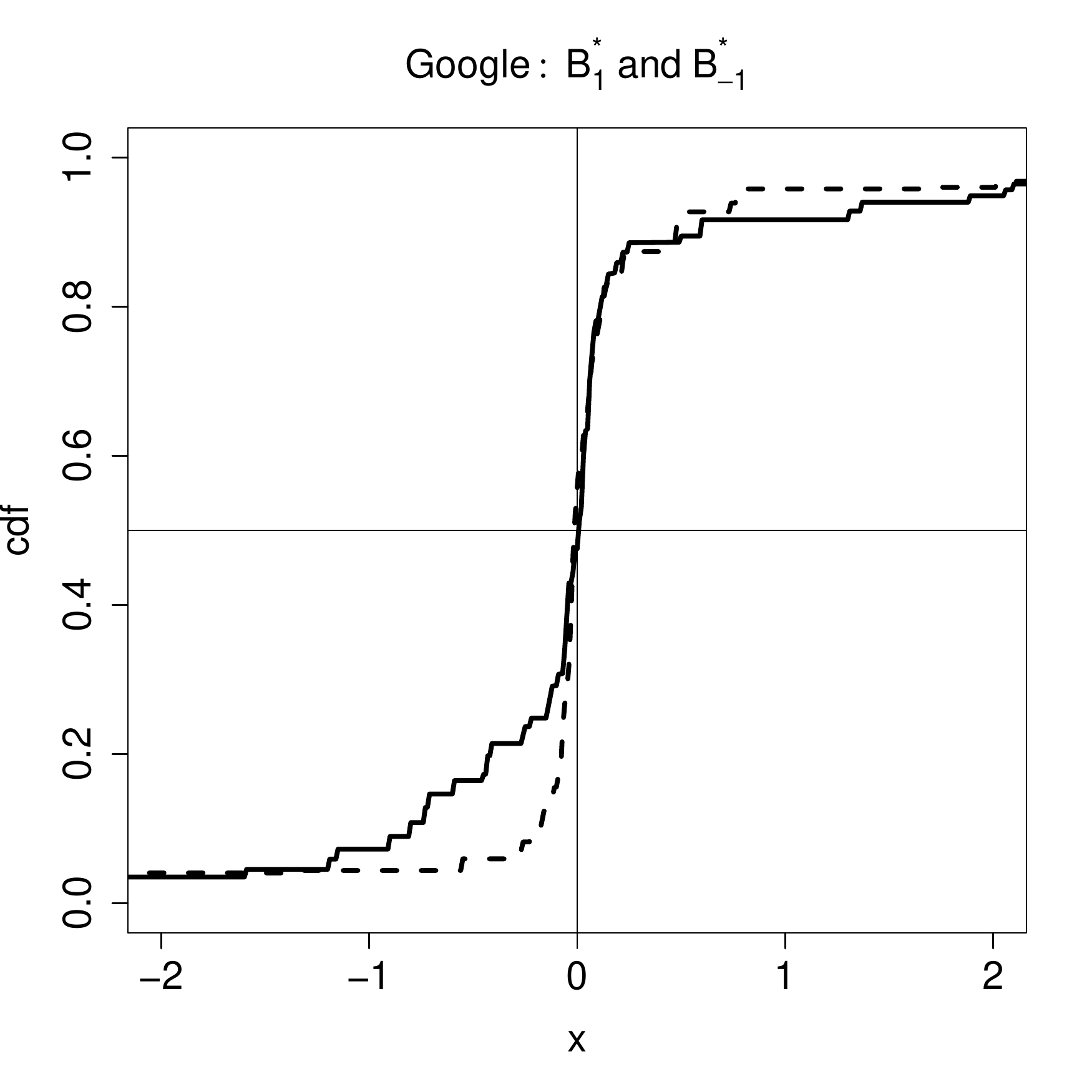}
\\
\includegraphics[width=0.4\textwidth]{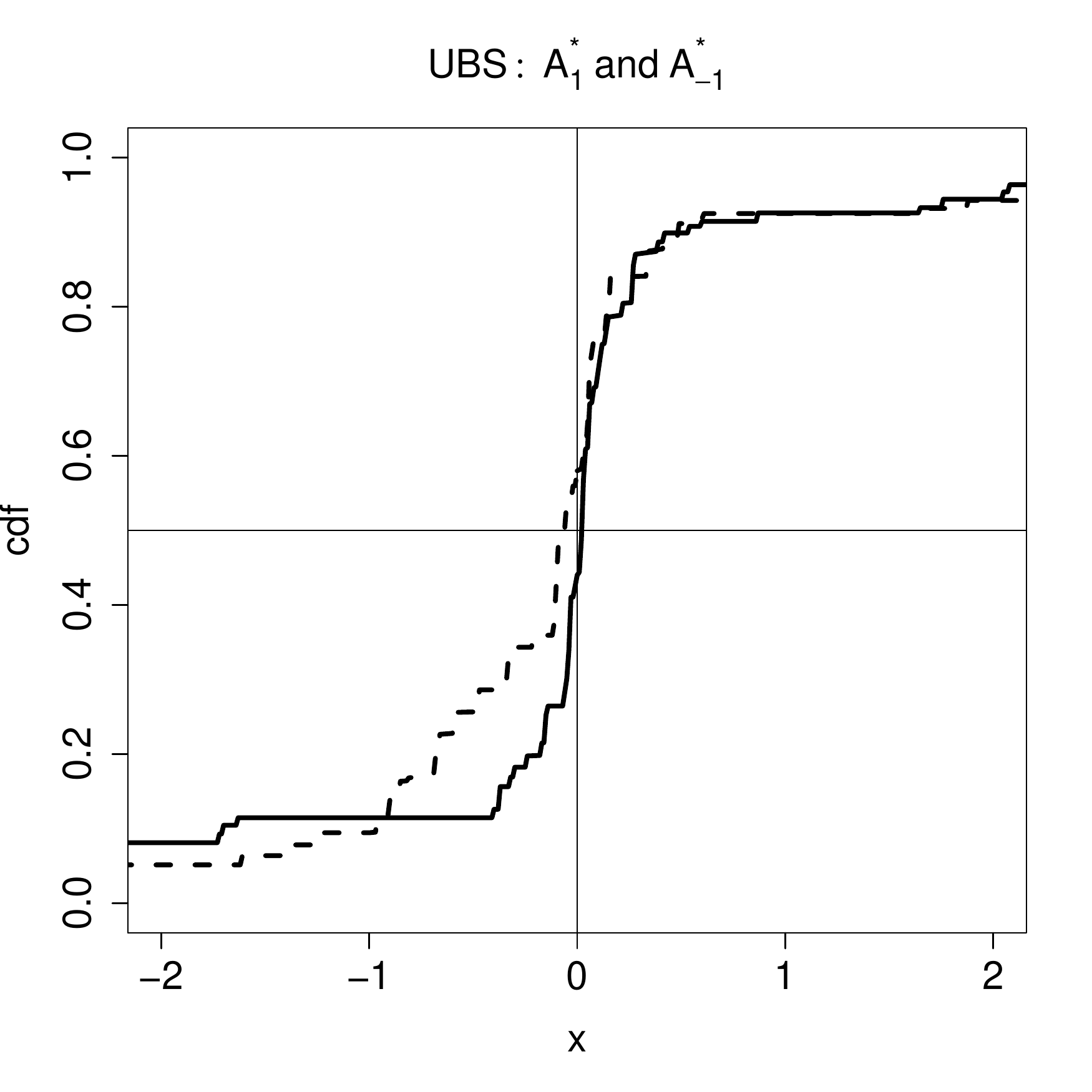}
&
\includegraphics[width=0.4\textwidth]{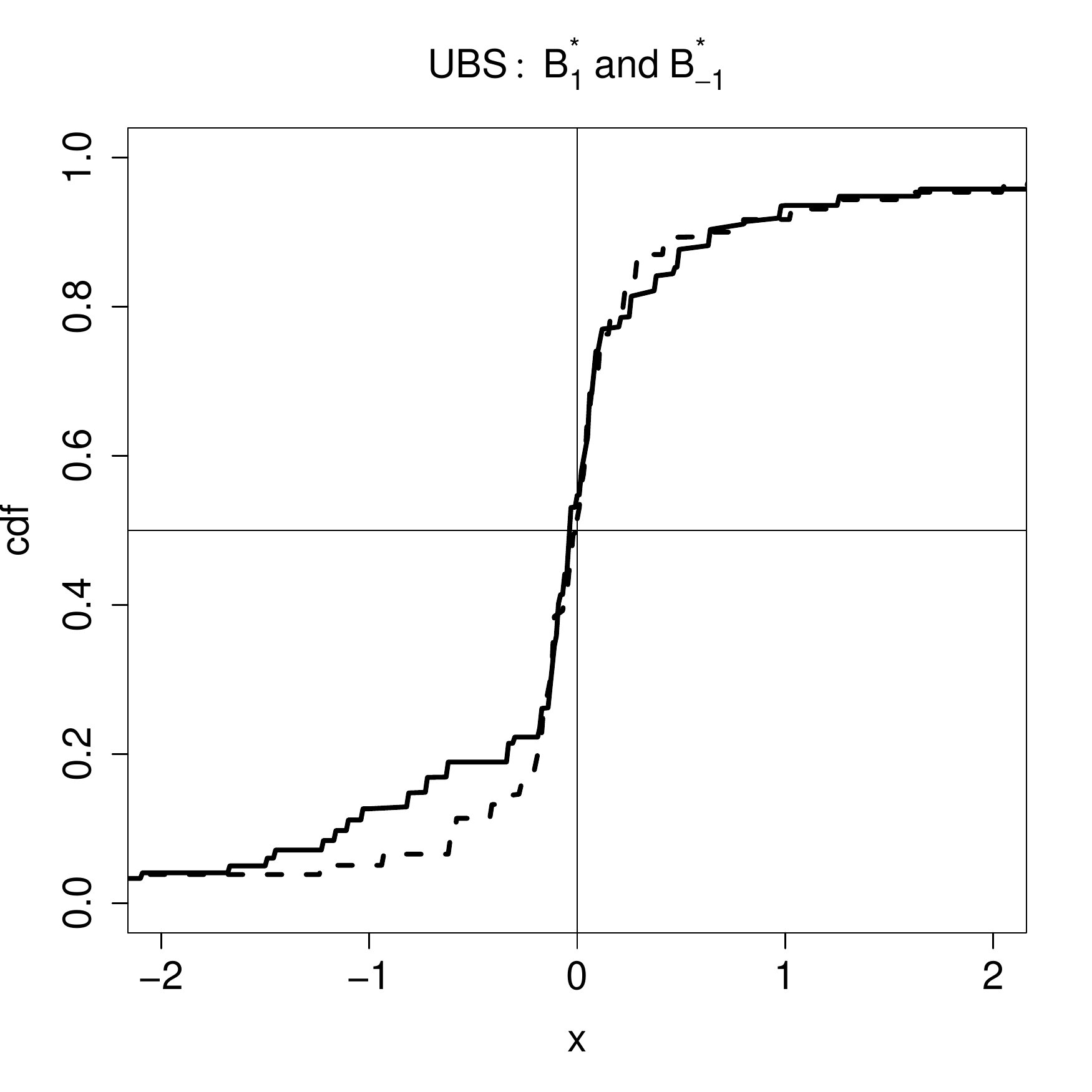}
\end{tabular}
\end{center}
\caption{\label{app6} \footnotesize
Estimated distribution functions of $A_1^*$ and $A_{-1}^*$ (left) and of $B_1^*$ and $B_{-1}^*$ (right) based on the monotonized mixture estimator for transformed log-returns on Google stock prices (top) and UBS stock prices (bottom). Solid lines: $A_1^*$ and $B_{1}^*$; dashed lines: $A_{-1}^*$ and $B_{-1}^*$.}
\end{figure}


\section{Appendix}
\label{appendix}

\subsection{Examples}
\label{sec:additional}

\begin{example}[The t-copula]\label{ex:tcopula}
In the context of the copula Markov model (Section~\ref{sec:examples:copula}), consider the bivariate t-copula defined by
\[
  C^t_{\nu,\rho}(u,v)
  = \int_{-\infty}^{t^{-1}_{\nu}(u)} \int_{-\infty}^{t^{-1}_{\nu}(v)}
  \frac{1}{2\pi (1 - \rho^2)^{1/2}}
  \left\lbrace
    1 + \frac{x^2-2\rho xy + y^2}{\nu \, (1-\rho^2)}
  \right\rbrace^{-(\nu+2)/2}
  \diff x \, \diff y,
\]
where $\rho\in(-1,1)$, $t_{\nu}$ is the cdf of the univariate t-distribution with $\nu>0$ degrees of freedom and $t^{-1}_{\nu}$ is the corresponding quantile function. Contrary to the Gaussian copula, the t-copula allows for an asymptotically non-negligible probability of joint extremes. The t-copula is exchangeable and radially symmetric. The partial derivative of $C^t_{\nu,\rho}(u,v)$ with respect to the first coordinate equals \citep[proof of Proposition~4]{Dematra2005}
\begin{equation}
\label{eq:cond:t:copula}
\frac{\partial}{\partial u} C^t_{\nu,\rho}(u,v)
=t_{\nu +1}\left(\left\lbrace \frac{t^{-1}_{\nu}\left(v\right)}{t^{-1}_{\nu}\left(u\right)}-\rho\right\rbrace\left(\frac{\nu +1}{1-\rho^2}\right)^{1/2}\left\lbrace 1+\frac{\nu}{\left(t^{-1}_{\nu}\left(u\right)\right)^2}\right\rbrace^{-1/2} \sign(t^{-1}_{\nu}\left(u\right))\right).
\end{equation}
The limits \eqref{eq:CA1pos} and \eqref{eq:CA1neg} can be calculated using symmetry of the t-distribution and the regular variation of $t_{\nu}$  with index $\nu$. Under the assumptions of Proposition~\ref{prop:copMarkov}, we obtain
\begin{equation}
\label{eq:A1:tcopula}
  \prob{ A_1 \le x } =
  \begin{cases}
    t_{\nu+1}\left(\left(x^{\alpha/\nu}-\rho\right)\left(\frac{\nu +1}{1-\rho^2}\right)^{1/2}\right) & \text{if $x \geq 0$,} \\[1em]
    t_{\nu+1}\left(\left(-(\frac{1-p}{p} \, \abs{x}^{-\alpha})^{-1/\nu}-\rho\right)\left(\frac{\nu +1}{1-\rho^2}\right)^{1/2}\right) & \text{if $x < 0$.}
  \end{cases}
\end{equation}
As the t-copula is exchangeable, we find $\law{A_1} = \law{A_{-1}}$ and $\law{B_1} = \law{B_{-1}}$. Moreover, by radial symmetry we have $\eta_{1,1} = 1 - \eta_{0,0}$ and $\eta_{1,0} = 1 - \eta_{0, 1}$, linking up the laws of $A_1$ and $B_1$ through \eqref{eq:A1:eta} and \eqref{eq:B1:eta}. In particular, if $p = 1/2$, then $\law{A_1} = \law{B_1}$.
\end{example}

\begin{example}[Extreme--value copulas]\label{ex:tEVcopula}
A bivariate copula  $C(u,v)$ is an extreme value copula if and only if it admits the representation
\[
  C(u,v)
  = \exp \left\{ D \left(\frac{\log v}{\log (uv)}\right) \log (uv)\right\},
  \qquad (u,v) \in (0,1]^2 \setminus \{(1,1)\}.
\]
The \emph{Pickands dependence function} $D : [0,1] \rightarrow [1/2,1]$ is convex and satisfies $\max(w, 1-w) \leq f(w) \leq 1$ for all $w \in [0,1]$. See~\citet{Gudendorf2010} for a survey of extreme-value copulas.

If $D$ is continuously differentiable with derivative $D'$, then, under the conditions of Proposition~\ref{prop:copMarkov},
\begin{equation*}
  \prob{ A_1 \le x } =
    \left\lbrace
      D\left(\frac{1}{x^{\alpha}+1}\right) -
      \frac{1}{x^{\alpha}+1} \, D'\left(\frac{1}{x^{\alpha}+1}\right)
    \right\rbrace
  \1(x \geq 0),
\end{equation*}
whereas $B_1$ and $B_{-1}$ are degenerate at $0$. In particular, $\prob{A_1 = 0} = 1 - D'(0)$. The law of $A_{-1}$ has the same form as the one of $A_1$ upon replacing $D$ by $w \mapsto D(1-w)$. If $D(w) = D(1-w)$ then $C$ is exchangeable and $\law{A_{-1}} = \law{A_1}$.

The following two parametric families are well known:
\begin{itemize}
\item
For the \emph{asymmetric logistic model} \citep{Tawn1988} with parameters $\theta\geq 1$ and $\psi_1,\psi_2\in (0,1]$, we have
\begin{align*}
  D(w)
  &= (1-\psi_1)w + (1-\psi_2)(1-w) +
  \{(\psi_1 w)^\theta+(\psi_2(1-w))^\theta\}^{1/\theta}, \\
  \prob{A_1 \le x}
  &= 1-\psi_2 +\psi_2\left\lbrace 1+\left(\psi_1 /(\psi_2 x^\alpha)\right)^{\theta}\right\rbrace^{(1-\theta)/\theta}, \qquad x \ge 0.
\end{align*}
The special case $\psi_1=\psi_2=1$ yields the Gumbel--Hougaard copula.
\item
For the \emph{asymmetric negative logistic model} \citep{Joe1990} with parameters $\theta>0$ and $\psi_1,\psi_2\in (0,1]$, we have
\begin{align*}
  D(w)
  &= 1-\{(\psi_1 w)^{-\theta}+(\psi_2(1-w))^{-\theta}\}^{-1/\theta}, \\
  \prob{A_1 \le x}
  &= 1-\psi_2\left\lbrace 1+\left(\psi_2 x^\alpha/\psi_1\right)^{\theta} \right\rbrace^{-(1+\theta)/\theta},
  \qquad x \ge 0.
\end{align*}
The special case $\psi_1 = \psi_2 = 1$ yields the Galambos copula.
\end{itemize}
\end{example}

\begin{example}[Stochastic recurrence equations] \label{exam:SRE}
We will show that stationary solutions to the stochastic recurrence equation \eqref{eq_SRE} with $(C_t,D_t)\in[0,\infty)^2$ satisfy the conditions (B) and (C) in Section~\ref{sec:CLT:known}. Asymptotic normality of the forward and backward estimators (Theorem~\ref{cor:asnorm}) follows if the cdf of $\law{A_1} = \law{C_1}$ is continuous on $[x_0,\infty)$.

We assume that the conditions in \citet{Kesten73} are fulfilled (see Section~\ref{sec:prelimin}). Then $(X_t)_{t\in\ZZ}$ is geometrically $\beta$-mixing, i.e., there exist constants $\eta\in(0,1)$ and $\tau>0$ such that $\beta_{n,k}\le \tau \eta^k$ \citep[Corollary 2.4.1]{doukhan95}. Therefore, condition (B) is satisfied with $l_n=2\log n/|\log\eta|$ and suitably chosen $r_n=o(\min\{(nv_n)^{1/2}, v_n^{-1}\})$, provided $v_n=o(1/\log n)$ and $(\log n)^2/n=o(v_n)$.

To establish condition (C), let  $\Pi_{i+1,j} := \prod_{k=i+1}^j C_k$ and $V_k:= \sum_{j=1}^k \Pi_{j+1,k} D_j$. Iterating~\eqref{eq_SRE} yields $X_k = V_k+\Pi_{1,k}X_0$. By independence of $(V_k, \Pi_{1,k})$ and $X_0$, one has
\begin{eqnarray*}
  \prob{ X_k>u_n\mid X_0>u_n}
  & \le & v_n^{-1} \prob{X_0>u_n, V_k>u_n/2} + v_n^{-1}\prob{X_0>u_n, \Pi_{1,k}X_0>u_n/2} \\
  & = & \prob{V_k>u_n/2} + v_n^{-1} \int_{u_n}^\infty \prob{\Pi_{1,k}> u_n/(2t)} \, \law{X_0}(\diff  t).
\end{eqnarray*}
There exists $\xi\in (0,\alpha)$ such that $\rho:=\texpec{C_1^\xi}<1$. Thus $\prob{\Pi_{1,k}>u_n/(2t)}\le \texpec{\Pi_{1,k}^\xi} (2t/u_n)^\xi = \rho^k  (2t/u_n)^\xi$, which in turn implies
\[
  v_n^{-1} \int_{u_n}^\infty \prob{\Pi_{1,k} >u_n/(2t)} \, \law{X_0}(\diff  t)
  \le \rho^k \expec{(2X_0/u_n)^\xi \mid X_0>u_n}
  \le 2^{\xi+1} \texpec{Y_0^\xi} \rho^k
\]
for all $k\in\NN$ and sufficiently large $n$. Because $\prob{V_k>u_n/2}\le \prob{X_k>u_n/2} \le 2^{1-\alpha} v_n$ for all $k\in\NN$ and sufficiently large $n$, one may conclude that, for some constant $c>0$,
\[
  \prob{X_k>u_n\mid X_0>u_n} \le c( v_n +  \rho^k) =: s_n(k) .
\]
Condition~(C) then follows from the fact that, as $n \to \infty$,
\[
  \sum_{k=1}^{r_n} s_n(k) = c r_n v_n + c \sum_{k=1}^{r_n} \rho^k \;\to\; c \sum_{k=1}^\infty \rho^k = \sum_{k=1}^\infty \lim_{n\to\infty} s_n(k) <\infty.
\]
\end{example}

\subsection{Proofs}
\label{sec:proofs}

\begin{proof}[Proof of Lemma~\ref{lem:standardization}]
Note that $\abs{X^*_t} \geq u$ if and only if $F_{\abs{X_0}}\left( \abs{X_t}\right)\geq 1-1 /u$ if and only if $\abs{X_t} \geq U(u) := F_{\abs{X_0}}^{\leftarrow}(1-1/u)$ with $F_{\abs{X_0}}^{\leftarrow}$ denoting the quantile function and $F_{\abs{X_0}}$ the distribution function of $\abs{X_t}$. Moreover, regular variation of $|X_0|$ implies regular variation of $U$ as well as $\prob{\abs{X_0} > u} / \prob{\abs{X_0} \ge u} \to 1$ and $u \, \bar{F}_{\abs{X_0}}(U(u))\to 1$ as $u\to\infty$. Thus, from weak convergence \eqref{eq:tailprocess} we may conclude
\begin{align*}
  \lefteqn{
    \law{ \left(\frac{X^*_i}{u}\right)_{s\le i\le t}\,\Big|\,\abs{X^*_0} \geq u }
  } \\
  &=
  \law{
    \left(
      \frac{X_i / \abs{X_i}}{u \, \bar{F}_{\abs{X_0}}( \abs{X_i} )}
      \cdot \1(X_i \ne 0)
    \right)_{s\le i\le t}
    \,\Big|\,
    \abs{ X_0} \geq U(u)
  } \\
  &=
  \law{
    \left(
      \frac{X_i /U(u)}{\abs{X_i /U(u)} } \cdot
      \1\left( \frac{X_i}{U(u)}\neq 0 \right)\cdot
      \frac{1}{u \, \bar{F}_{\abs{X_0}}(U(u))} \cdot
      \frac%
	{\bar{F}_{\abs{X_0}}(U(u))}%
	{\bar{F}_{\abs{X_0}}\left(U(u)\frac{ \abs{ X_i}}{U(u)}\right)}
    \right)_{s\le i\le t}
    \,\Bigg|\, \abs{X_0} \geq U( u)
  } \\
  &\dto
  \law{
    \left(
      \frac{Y_i}{\abs{Y_i}} \cdot
      \1(Y_i \neq 0) \cdot
      1 \cdot
      \frac{1}{\abs{ Y_i}^{-\alpha}}
    \right)_{s\le i\le t}
  },
  \qquad u\rightarrow\infty,
\end{align*}
where in the last step we applied the extended continuous mapping theorem 1.11.1 of \cite{vdVW96}.  Therefore, $(X_t^*)_{t \in \ZZ}$ has a tail process $(Y_t^*)_{t \in \ZZ}$ with $Y_t^*= \sign(Y_t)\abs{Y_t}^{\alpha}$ and is thus regular varying with index $\alpha^*=1$. Moreover, $\Theta_t^*= Y_t^*/|Y_0^*|= \sign(\Theta_t)\abs{\Theta_t}^{\alpha}$.
\end{proof}

\begin{proof}[Proof of Proposition~\ref{prop:copMarkov}.]
Let $g$ be the density function of $G$. For real $x$ and $u$,
\[
  \prob{ X_1/X_0 \le x \mid X_0 > u }
  =
  \int_u^\infty
    \prob{ X_1/X_0 \le x \mid X_0 = y } \,
    \frac{g(y)}{1-G(u)} \,
  \differ y.
\]
As a consequence, $\lim_{u \to \infty} \prob{ X_1/X_0 \le x \mid X_0 > u } = \lim_{y \to \infty} \prob{ X_1/X_0 \le x \mid X_0 =y}$, provided the latter limit exists.

We compute the conditional distribution of $X_1$ given $X_0 = y$. For $(x_0, x_1) \in \reals^2$, as $G(-\infty) = 0$ and $C(0, \, \cdot \, ) = 0$, we have
\begin{align*}
  \prob{ X_0 \le x_0, X_1 \le x_1 }
  &= C \bigl( G(x_0), G(x_1) \bigr) - C \bigl( G(-\infty), G(x_1) \bigr) \\
  &= \int_{-\infty}^{x_0} \dot{C}_1 \bigl( G(y), G(x_1) \bigr) \, g(y) \, \differ y \\
  &= \expec{ \1\{ X_0 \le x_0 \} \, \dot{C}_1 \bigl( G(X_0), G(x_1) \bigr) }.
\end{align*}
We find that a version of the conditional distribution of $X_1$ given $X_0 = y$ is given by
\begin{equation}
\label{eq:dotC1}
  \prob{ X_1 \le x_1 \mid X_0 = y } = \dot{C}_1 \bigl( G(y), \, G(x_1) \bigr),
  \qquad x_1 \in \reals.
\end{equation}

Assume $p > 0$. As $s \searrow 0$, the functions $[0, \infty) \to \reals : z \mapsto \dot{C}_1(1-s, 1-sz)$ are assumed to converge pointwise to a continuous limit. Since  these functions are monotone, the convergence holds locally uniformly, i.e., if $0 \le z(s) \to z$ as $s \searrow 0$, then $\lim_{s \searrow 0} \dot{C}_1(1 - s, 1 - s \, z(s)) = \lim_{s \searrow 0} \dot{C}_1(1 - s, 1 - s \, z) = \eta_{1,1}(z)$.

Moreover, the function $\bar{G} = 1 - G$ is regularly varying at infinity with index $-\alpha$, too. Indeed, as $u \to \infty$ we have $\prob{ X_0 > u } / \prob{ \abs{X_0} > u } = p > 0$, whereas the function $u \mapsto \prob{ \abs{X_0} > u }$ is regularly varying at infinity with index $-\alpha$.

For $x \in (0, \infty)$, we find
\begin{align*}
  \prob{ X_1/X_0 \le x \mid X_0 = y }
  &= \dot{C}_1 \bigl( G(y), \, G(xy) \bigr) \\
  &=
  \dot{C}_1
  \left(
    1 - \bar{G}(y), \, 1 - \bar{G}(y) \, \frac{\bar{G}(xy)}{\bar{G}(y)}
  \right) \\
  &\to \lim_{s \searrow 0} \dot{C}_1(1 - s, \, 1 - s \, x^{-\alpha})
  = \eta_{1,1}(x^{-\alpha}),
  \qquad y \to \infty.
\end{align*}
Similarly,  $\lim_{y \to \infty} G(-y)/\bar{G}(y) = (1-p)/p$ implies, for $x \in (-\infty, 0)$,
\begin{align*}
  \prob{ X_1/X_0 \le x \mid X_0 = y }
  &= \dot{C}_1 \bigl( G(y), \, G(-\abs{x}y) \bigr) \\
  &=
  \dot{C}_1
  \left(
    1 - \bar{G}(y), \,
    \bar{G}(y) \,
    \frac{G(-\abs{x}y)}{\bar{G}(\abs{x}y)} \,
    \frac{\bar{G}(\abs{x}y)}{\bar{G}(y)}
  \right) \\
  &\to
  \lim_{s \searrow 0}
  \dot{C}_1 \bigl( 1 - s, \, s \, [(1-p)/p] \, \abs{x}^{-\alpha} \bigr)
  = \eta_{1,0}([(1-p)/p] \, \abs{x}^{-\alpha}),
\end{align*}
as $y \to \infty$. We conclude that $\mathcal{L}(X_1/X_0 \mid X_0 = y)$ converges weakly, as $y \to \infty$, to the distribution $\mathcal{L}(A_1)$ given by~\eqref{eq:A1:eta}. By the argument at the beginning of the proof, the same then holds true for $\mathcal{L}(X_1/X_0 \mid X_0 > u)$ as $u \to \infty$.

The proof of \eqref{eq:B1:eta} is entirely similar.
\end{proof}

\begin{proof}[Proof of Lemma \ref{lem:tailsign}]
For real $x$, write $x^+ = \max(x, 0)$ and $x^- =  \max(-x, 0)$.
Recall the time-change formula, \eqref{eq:timechange}. Setting $s = 0$, $t = 1$, and $i = 1$ yields, for every measurable function $f : \reals^2 \to \reals$ such that $f(0, \, \cdot \,) = 0$ and for which the expectations below exist,
\begin{equation}
\label{eq:timechange:1}
  \expec{ f(\Theta_{-1}, \Theta_0) }
  = \expec{ f \left( \frac{\Theta_0}{\abs{\Theta_1}}, \frac{\Theta_1}{\abs{\Theta_1}} \right) \, \abs{\Theta_1}^\alpha \, \1( \Theta_1 \ne 0 ) }.
\end{equation}
Setting $s = -1$, $t = 0$ and $i = -1$ yields the same formula, but with the roles of $\Theta_{-1}$ and $\Theta_1$ interchanged.
By the time-change formula~\eqref{eq:timechange:1},
\begin{align*}
  \expec{ (\Theta_1^+)^\alpha }
  &= \prob{ \Theta_0 = + 1, \, \Theta_{-1} \ne 0 }, &
  \expec{ (\Theta_1^-)^\alpha }
  &= \prob{ \Theta_0 = - 1, \, \Theta_{-1} \ne 0 }.
\end{align*}
[The above equations even hold without conditions~(i) and~(ii).] Adding both identities yields, in view of~(i),
\[
  \prob{ \Theta_{-1} \ne 0 } = \expec{ \abs{\Theta_1}^\alpha } = 1
\]
and thus $\expec{ (\Theta_1^+)^\alpha } = \prob{ \Theta_0 = + 1 }=p$.

For $M := \Theta_1 / \Theta_0$, let $\mu_+ = \expec{ (M^+)^\alpha }$ and $\mu_- = \expec{ (M^-)^\alpha }$. Then
\[
  \mu_+ + \mu_- = \expec{ \abs{M}^\alpha } = \expec{ \abs{\Theta_1}^\alpha } = 1.
\]
By (ii), we have
\begin{align*}
  p
  = \expec{ (\Theta_1^+)^\alpha } 
  = \expec{ ((M \Theta_0)^+)^\alpha } 
  &= p \, \expec{ (M^+)^\alpha } + (1-p) \, \expec{ (M^-)^\alpha } \\
  &= p \, (1-\mu_-) + (1-p) \, \mu_-.
\end{align*}
After simplification, we find
\[
  (1 - 2p) \, \mu_- = 0,
\]
so that either $p = 1/2$ or $\mu_- = 0$, i.e., $\prob{ M < 0 } = 0$. This yields the first statement.

Second, we show that $\Theta_{-1}/\Theta_0$ is independent of $\Theta_0$, too. If $p \in \{0, 1\}$, then $\Theta_0$ is degenerate; without loss of generality, assume that $0 < p < 1$. We need to show that, for bounded measurable functions $g : \reals \to \reals$,
\[
  \expec{ g(\Theta_{-1}/\Theta_0) \mid \Theta_0 = +1 }
  =
  \expec{ g(\Theta_{-1}/\Theta_0) \mid \Theta_0 = -1 }.
\]
Equivalently, we need to show that, for such $g$,
\begin{equation}
\label{eq:indep:tocheck}
  (1-p) \, \expec{ g(\Theta_{-1}/\Theta_0) \, \1(\Theta_0 = +1) }
  =
  p \, \expec{ g(\Theta_{-1}/\Theta_0) \, \1(\Theta_0 = -1) }.
\end{equation}
The above formula clearly holds for constant functions $g$. Hence, it holds for a function $g$ as soon as it holds for the function $y \mapsto g(y) - g(0)$. Without loss of generality, we may therefore assume that $g(0) = 0$. But then, by the time-change formula \eqref{eq:timechange:1} applied to the two functions $f_\pm(\theta_0, \theta_1) = g(\theta_0/\theta_1) \, \1(\theta_1 = \pm 1)$ and using independence of $M = \Theta_1/\Theta_0$ and $\Theta_0$,
\begin{eqnarray*}
  \expec{ g(\Theta_{-1}/\Theta_0) \, \1(\Theta_0 = +1) }
  &=& \expec{ g(\Theta_0/\Theta_1) \, (\Theta_1^+)^\alpha } \\
  &=&
  \expec{ g(\Theta_0/\Theta_1) \, ((\Theta_1/\Theta_0)^+)^\alpha \, \1(\Theta_0 = +1) } \\
  && \mbox{} +
  \expec{ g(\Theta_0/\Theta_1) \, (-(\Theta_1/\Theta_0)^+)^\alpha \, \1(\Theta_0 = -1) }
  \\
  &=&
  p \, \expec{ g(1/M) \, (M^+)^\alpha) } +
  (1-p) \, \expec{ g(1/M) \, (M^-)^\alpha }, \\
  \expec{ g(\Theta_{-1}/\Theta_0) \, \1(\Theta_0 = -1) }
  &=&
  p \, \expec{ g(1/M) \, (M^-)^\alpha) } +
  (1-p) \, \expec{ g(1/M) \, (M^+)^\alpha }.
\end{eqnarray*}
Since either $p = 1/2$ or $\prob{M < 0} = 0$, equation~\eqref{eq:indep:tocheck} follows.
\end{proof}

\begin{proof}[Proof of Lemma~\ref{lem:tc:AB}.]
To show \eqref{eq:A1:tc:pos}, apply the time-change formula \eqref{eq:timechange} with $s=-1$, $t=0$, $i=-1$, and $f(y_{-1},y_0)=\1(y_0 / y_{-1} > x,\; y_{-1}=1)$, where $x \ge 0$. It follows that
\begin{align*}
  \prob{A_1>x}
  &= \expec{ \1\left(\frac{\Theta_1}{ \Theta_{0} } > x\right) \1\left(\Theta_{0}=1\right) } / \prob{ \Theta_0 = 1} \\
  &= \expec{ \1\left(\frac{\Theta_0}{ \Theta_{-1} } > x\right) \1\left(\Theta_{-1}>0\right)\vert \Theta_{-1}\vert^{\alpha} } / \prob{ \Theta_0 =1 } \\
  &= \expec{ A_{-1}^{\alpha} \, \1(1/A_{-1} > x)  }.
\end{align*}
To show \eqref{eq:A1:tc:neg}, apply the time-change formula to the function $f(y_{-1}, y_0) = \1(y_0 / y_{-1} \leq x, \; y_{-1}=1)$. The proofs of \eqref{eq:B1:tc:pos} and \eqref{eq:B1:tc:neg} are similar.
\end{proof}

Next we establish the asymptotic normality of the forward and backward estimator, using the limit theory for empirical processes of cluster functionals developed by \cite{drees2010limit}. to this end, for $x\ge 0$, define functions $\phi_1,\phi_{2,x},\phi_{3,x}:[0,\infty)^3\to [0,\infty)$ by
\begin{eqnarray*}
  \phi_1(y_{-1},y_0,y_1) & := & \1(y_0>1)\\
  \phi_{2,x}(y_{-1},y_0,y_1) & := & \1(y_1/y_0>x,y_0>1)\\
  \phi_{3,x}(y_{-1},y_0,y_1) & := & (y_{-1}/y_0)^\alpha \1(y_0/y_{-1}>x,y_{-1}>0,y_0>1).
\end{eqnarray*}
The forward and backward estimators of the cdf of $A_1$ can be written as
\begin{align*}
  \CDFfA{x} &= 1-\frac{\sum_{i=1}^n \phi_{2,x}(X_{n,i})}{\sum_{i=1}^n \phi_1(X_{n,i})}, &
  \CDFbA{x} &= 1-\frac{\sum_{i=1}^n \phi_{3,x}(X_{n,i})}{\sum_{i=1}^n \phi_1(X_{n,i})}
\end{align*}
with $X_{n,i}$ defined in \eqref{eq:Xnidef}.

Taking up the notation of \cite{drees2010limit}, we consider the empirical process $\tilde{Z}_n$ defined by
\[
  \tilde Z_n(\psi) := (nv_n)^{-1/2} \sum_{i=1}^n \bigl( \psi(X_{n,i})-\expec{\psi(X_{n,i})} \bigr),
\]
where $\psi$ is one of $\phi_1$, $\phi_{2,x}$ or $\phi_{3,x}$.

\begin{proposition}
\label{th:procconv}
  Suppose that $(X_t)_{t\in\ZZ}$ is a stationary, regularly varying time series and that the conditions (A($x_0$)), (B),  and (C) are fulfilled for some $x_0\ge 0$. Then, for all $y_0\in [x_0,\infty)\cap(0,\infty)$, the process $\big(\tilde Z_n(\phi_1),(\tilde Z_n(\phi_{2,x}))_{x\in [x_0,\infty)}, (\tilde Z_n(\phi_{3,y}))_{y\in [y_0,\infty)}\big)$ converges weakly to a centered Gaussian process $\tilde Z$ with covariance function given by
  \begin{multline}
          \cov{\tilde Z(\psi_1),\tilde Z(\psi_2)} \\
         \shoveleft =  \expec{\psi_1(Y_{-1},Y_0,Y_1) \, \psi_2(Y_{-1},Y_0,Y_1)} + \sum_{k=1}^\infty \Big(\expec{\psi_1(Y_{-1},Y_0,Y_1) \, \psi_2(Y_{k-1},Y_k,Y_{k+1})} \\
        +
        \expec{\psi_2(Y_{-1},Y_0,Y_1) \, \psi_1(Y_{k-1},Y_k,Y_{k+1})}\Big)
          \label{eq:tildeZcov}
  \end{multline}
  for all $\psi_1,\psi_2\in\{\phi_1, \phi_{2,x}, \phi_{3,y} : x\ge x_0, y\ge y_0\}$.
\end{proposition}
\begin{proof}
  We argue similarly as  in the proof of Corollary~3.6 and Remark~3.7 of \cite{drees2010limit}: we first establish weak convergence of all finite dimensional distributions using Theorem~2.3 of that paper, and then the asymptotic equicontinuity of the processes $(\tilde Z_n(\phi_{2,x}))_{x \ge x_0}$ and $(\tilde Z_n(\phi_{3,y}))_{y \ge y_0}$ by applying Theorem~2.10, from which the assertion follows.

  First we verify that conditions (C1)--(C3) of \cite{drees2010limit} are fulfilled so that Theorem~2.3 on the convergence of finite-dimensional distributions applies. As in the proof of Corollary~3.6, (C1) can be derived from condition~(3.5) of that paper, which in turn is an easy consequence of condition~(C), because by the stationarity of $(X_t)_{t\in\ZZ}$, we have
  \begin{eqnarray*}
    \expec{\left(\sum_{i=1}^{r_n} \1(X_{n,i}\ne 0)\right)^2}
    & = & \sum_{i=1}^{r_n} \sum_{j=1}^{r_n} \prob{X_i>u_n,X_j>u_n} \\
    & \le & 2 r_nv_n \sum_{k=0}^{r_n-1} \Big(1-\frac{k}{r_n}\Big) \prob{X_k>u_n\mid X_0>u_n} \\
    & \le & 2r_nv_n \sum_{k=0}^{r_n-1} s_n(k)=O(r_nv_n), \qquad n \to \infty.
  \end{eqnarray*}
  Since all functions $\phi_1$, $\phi_{2,x}$ and $\phi_{3,y}$ are bounded and since $r_n=o((nv_n)^{1/2})$, condition~(C2) is obviously fulfilled.

  For the convergence of all finite-dimensional marginal distributions it remains to establish condition (C3) of \cite{drees2010limit}, i.e.,
  \begin{equation}  \label{eq:covconv}
    \frac 1{r_nv_n} \cov{\sum_{i=1}^{r_n} \psi_1(X_{n,i}),\sum_{j=1}^{r_n} \psi_2(X_{n,j})}
    \;\to\; \cov{\tilde Z(\psi_1),\tilde Z(\psi_2)},
    \qquad n \to \infty,
  \end{equation}
  for all $\psi_1,\psi_2\in\{\phi_1, \phi_{2,x}, \phi_{3,y} : x \ge x_0, y\ge y_0\}$. Similarly as above, by the stationarity of $(X_t)_{t\in\ZZ}$, the left-hand side equals
  \begin{multline*}
    \frac 1{r_n v_n}
    \expec{\sum_{i=1}^{r_n} \sum_{j=1}^{r_n} \psi_1(X_{n,i}) \, \psi_2(X_{n,j})}
    + O(r_nv_n)\\
    \shoveleft = \expec{ \psi_1(X_{n,0}) \, \psi_2(X_{n,0})\mid X_0>u_n}\\
    +
     \sum_{k=1}^{r_n-1} \Bigl(1-\frac{k}{r_n}\Bigr)
    \Bigl(
      \expec{ \psi_1(X_{n,0}) \, \psi_2(X_{n,k}) \mid X_0 > u_n} +
      \expec{ \psi_1(X_{n,k}) \, \psi_2(X_{n,0}) \mid X_0 > u_n}
    \Bigl).
  \end{multline*}
 By assumption (A($x_0$)), all functions under consideration are a.s.\ continuous and bounded. The definition of the tail process then yields
 \[
  \expec{ \psi_1(X_{n,0}) \, \psi_2(X_{n,k})\mid X_0>u_n}\;\to\; \expec{\psi_1(Y_{-1},Y_0,Y_1) \, \psi_2(Y_{k-1},Y_k,Y_{k+1})},
 \qquad n \to \infty,
 \]
 for all $k\ge 0$ and all $\psi_1,\psi_2\in\{\phi_1,\phi_{2,x}, \phi_{3,y} : x\ge x_0, y\ge y_0\}$. Thus \eqref{eq:covconv} follows by Pratt's lemma \citep{pratt60} and condition~(C), since
 \begin{eqnarray*}
  \Big(1-\frac{k}{r_n}\Big) \expec{ \psi_1(X_{n,0}) \, \psi_2(X_{n,k})\mid X_0>u_n} &\le & \max(1,\tilde x_0^{-2\alpha}) \prob{X_k>u_n\mid X_0>u_n}\\
  &\le & \max(1,\tilde x_0^{-2\alpha}) \, s_n(k).
  \end{eqnarray*}

 In the second step, the asymptotic equicontinuity of $(\tilde Z_n(\phi_{2,x}))_{x\in [x_0,\infty)}$ and $(\tilde Z_n(\phi_{3,y}))_{y\in [y_0,\infty)}$ follows from Theorem 2.10 of \cite{drees2010limit} if their conditions (D1), (D2'), (D3), (D5) and (D6) are verified. Note that (D1) is obvious, that (D5) is an immediate consequence of the separability of the processes, and that  (D2') follows from $r_n=o((nv_n)^{1/2})$ and the boundedness of all functionals $\phi_{2,x}$ and $\phi_{3,y}$. Moreover, because the maps $x\mapsto \phi_{2,x}(y_{-1},y_0,y_1)$ and $x\mapsto \phi_{3,x}(y_{-1},y_0,y_1)$ are decreasing, condition (D6) can be concluded in the same way as in the case $d=1$ of Example~3.8 in \cite{drees2010limit}.

 It remains to establish the continuity condition (D3) for the semi-norm generated by the cdf $F^{(A_1)}$, i.e.,
 \begin{equation}
  \label{eq:contphi2}
   \lim_{\delta\downarrow 0} \limsup_{n\to\infty}
    \sup_{\substack{y>x\ge x_0 \\ F^{(A_1)}(y)-F^{(A_1)}( x)\le\delta}}
    \frac{1}{r_n v_n}
    \expec{
      \left(
	\sum_{i=1}^{r_n} \big(\phi_{2,y}(X_{n,i})-\phi_{2, x}(X_{n,i})\big)
      \right)^2
    }
   = 0
 \end{equation}
 and an analogous condition for $\phi_{3,y}$. By the usual stationarity argument, the expectation on the left-hand side can be bounded by a multiple of
 \[
    r_nv_n \sum_{k=0}^{r_n-1}
    \prob {X_1/ X_0\in(x,y], \, X_{k+1}/ X_k \in (x,y], \, X_k > u_n \mid X_0>u_n}.
  \]
 Now, for all fixed $M>0$, as $n \to \infty$,
 \begin{multline}
  \sum_{k=0}^M \prob {\frac{X_1}{X_0}\in(x,y],\frac{X_{k+1}}{X_k}\in(x,y], X_k>u_n\;\Big|\; X_0>u_n} \le  (M+1) \prob{\frac{X_1}{X_0}\in(x,y] \;\Big|\; X_0>u_n}
  \\
    \to  (M+1) \big(F^{(A_1)}(y)-F^{(A_1)}(x)\big)\le (M+1)\delta,
     \label{eq:bound1}
 \end{multline}
 uniformly for all $y> x\ge x_0$ by condition~(A($x_0$)). Moreover, condition~(C) implies
 \begin{multline}
  \sum_{k=M+1}^{r_n} \prob{\frac{X_1}{X_0}\in(x,y],\frac{X_{k+1}}{X_k}\in(x,y], X_k>u_n\;\Big|\; X_0>u_n}
  \\
   \le  \sum_{k=M+1}^{r_n} s_n(k)
   \;\to \; \sum_{k=M+1}^{\infty} \lim_{n\to\infty} s_n(k) <\infty.
    \label{eq:bound2}
  \end{multline}
  By choosing $M$ sufficiently large, the right-hand side of \eqref{eq:bound2} can be made arbitrarily small. Given such $M$, by choosing $\delta$ small, the right-hand side of \eqref{eq:bound1} can be made arbitrarily small too. Equation~\eqref{eq:contphi2} follows. The convergence statement~\eqref{eq:contphi2} involving the functions $(\phi_{3,y})_{y\ge y_0}$ can be concluded in a similar way.
\end{proof}

\begin{proof}[Proof of Theorem \ref{cor:asnorm}]
  For notational simplicity, we write $\bar F$ for $\bar F^{(A_1)}=1-F^{(A_1)}$ and use the notations $\FfA{x}:=1-\CDFfA{x}$ and $\FbA{y}:=1-\CDFbA{y}$ for the estimators of the survival function.
  Check that
  \begin{eqnarray*}
   \lefteqn{
    (nv_n)^{1/2}
    \left(
      \FfA{x} - \prob{X_1/X_0>x\mid X_0>u_n}
    \right)
    }\\
   & = & (nv_n)^{1/2}
  \left(
    \frac{\sum_{i=1}^n \phi_{2,x}(X_{n,i})}{\sum_{i=1}^n \phi_1(X_{n,i})}
    - \frac{\prob{X_1/X_0>x, X_0>u_n} }{v_n}
  \right) \\
   & = & (nv_n)^{1/2}
  \left(
    \frac{n \expec{ \phi_{2,x}(X_{n,0})}+(nv_n)^{1/2}\tilde Z_n(\phi_{2,x}) }{nv_n +(nv_n)^{1/2}\tilde Z_n( \phi_1)}
    -
    \frac{n \expec{ \phi_{2,x}(X_{n,0})}}{nv_n}
  \right) \\
   & = & \frac{\tilde Z_n(\phi_{2,x})- \prob{X_1/X_0>x\mid X_0>u_n}\tilde Z_n( \phi_1)}{1+(nv_n)^{-1/2} \tilde Z_n( \phi_1)},
  \end{eqnarray*}
  and likewise
  \begin{multline*}
   (nv_n)^{1/2}
  \left(
    \FbA{y}  - \expec{(X_{-1}/X_0)^\alpha \, \1(X_0/X_{-1}>y)\mid X_0>u_n}
  \right)
  \\
   =
  \frac{\tilde Z_n(\phi_{3,y})-  \expec{(X_{-1}/X_0)^\alpha \, \1(X_0/X_{-1}>y)\mid X_0>u_n}\tilde Z_n( \phi_1)}{1+(nv_n)^{-1/2} \tilde Z_n( \phi_1)}.
  \end{multline*}
  Hence, in view of \eqref{eq:bias1} and  \eqref{eq:bias2}, Proposition \ref{th:procconv} implies \eqref{eq:estconv} with
  $Z^{(f,A_1)}(x) =  \bar F(x)\tilde Z(\phi_1)-\tilde Z(\phi_{2,x})$ and $Z^{(b,A_1)}(y) =  \bar F(y)\tilde Z(\phi_1)-\tilde Z(\phi_{3,y})$.

  The covariance structure of the limiting process follows by direct calculations. Nonnegativity of $X_0$ implies $Y_0 > 1$ a.s. Moreover, as $(\Theta_t)_{t \in \ZZ}$ is a Markov spectral tail chain, the random variables $Y_k/Y_{k-1}=\Theta_k/\Theta_{k-1}$, $k\in\ZZ$, are independent; for $k \in \NN$, their common survival function is $\bar F$. From Proposition~\ref{th:procconv},  we obtain
  \begin{eqnarray*}
    \var(\tilde Z(\phi_1)) & = & 1+ 2\sum_{k=1}^\infty \prob{Y_k>1}\\
    \cov{\tilde Z(\phi_1),\tilde Z(\phi_{2,x})} & = & \prob{\Theta_1>x} + \sum_{k=1}^\infty \big( \prob{\Theta_{k+1}/\Theta_k>x,Y_k>1}+ \prob{\Theta_1>x,Y_k>1}\big)\\
    & = & \bar F(x)\sum_{k=0}^\infty \prob{Y_k>1}+ \sum_{k=1}^\infty\prob{\Theta_1>x,Y_k>1}, \\
    \cov{\tilde Z(\phi_{2,x}),\tilde Z(\phi_{2,y})}
  & = &
      \prob{\Theta_1>\max(x,y)} + \sum_{k=1}^\infty
    \big(\prob{\Theta_1>x,\Theta_{k+1}/\Theta_k>y,Y_k>1} \\
  & & \hspace*{4cm} \mbox{} +\prob{\Theta_1>y,\Theta_{k+1}/\Theta_k>x,Y_k>1}\big)\\
    & = & \bar F(\max(x,y)) + \sum_{k=1}^\infty \big(\prob{\Theta_1>x,Y_k>1}\bar F(y) \\
  & & \hspace*{4cm} \mbox{} + \prob{\Theta_1>y,Y_k>1}\bar F(x)\big).
  \end{eqnarray*}
  Similarly,
  \begin{align*}
    \cov{\tilde Z(\phi_1),\tilde Z(\phi_{3,x})}
    & =
    \bar F(x) \sum_{k=0}^\infty \prob{Y_k>1} + \sum_{k=1}^\infty \expec{(\Theta_{k-1}/\Theta_k)^\alpha \1(\Theta_k/\Theta_{k-1}>x,Y_k>1)}\\
     \cov{\tilde Z(\phi_{3,x}),\tilde Z(\phi_{3,y})}
    & =
     \bar F(\max(x,y)) + \sum_{k=1}^\infty \Big(\bar F(x) \expec{(\Theta_{k-1}/\Theta_k)^\alpha \1(\Theta_k/\Theta_{k-1}>y,Y_k>1)}\\
      &  \hspace*{3cm} \mbox{} + \bar F(y) \expec{(\Theta_{k-1}/\Theta_k)^\alpha \1(\Theta_k/\Theta_{k-1}>x,Y_k>1)} \Big)\\
    \cov{\tilde Z(\phi_{2,x}),\tilde Z(\phi_{3,y})}
    & =
     \bar F(x)\bar F(y)\sum_{k=0}^\infty \prob{Y_k>1} \\
     &  \hspace*{1cm} \mbox{} +
      \sum_{k=1}^\infty  \expec{(\Theta_{k-1}/\Theta_k)^\alpha \1(\Theta_1>x,\Theta_k/\Theta_{k-1}>y,Y_k>1)}.
  \end{align*}
  The asymptotic covariance functions of the forward and the backward estimators can then be derived as follows:
   \begin{eqnarray*}
    \lefteqn{\cov{Z^{(f,A_1)}(x) ,Z^{(b,A_1)}(y)}}\\
     & = & \cov{\tilde Z(\phi_{2,x}),\tilde Z(\phi_{3,y})}
    - \bar F(x) \cov{\tilde Z(\phi_1),\tilde Z(\phi_{3,y})}
    - \bar F(y) \cov{\tilde Z(\phi_1),\tilde Z(\phi_{2,x})} \\
    & & \mbox{} + \bar F(x) \, \bar F(y) \var(\tilde Z(\phi_1))\\
    & = &  \bar F(x) \, \bar F(y) \sum_{k=1}^\infty \prob{Y_k>1} - \bar F(x)
     \sum_{k=1}^\infty \expec{(\Theta_{k-1}/\Theta_k)^\alpha \, \1(\Theta_k/\Theta_{k-1}>y,Y_k>1)} \\
     & & \mbox{} - \bar F(y)   \sum_{k=1}^\infty \prob{\Theta_1>x,Y_k>1} + \sum_{k=1}^\infty
    \expec{
      (\Theta_{k-1}/\Theta_k)^\alpha \,
      \1(\Theta_1>x,\Theta_k/\Theta_{k-1}>y,Y_k>1)
    }.
    \end{eqnarray*}
    The other covariances can be calculated in a similar way.
  \end{proof}

\begin{proof}[Proof of Lemma~\ref{lemma:alphahatconv}]
  By similar arguments as used in the proof of Proposition~\ref{th:procconv}, one can show that under the present conditions the conclusion of Proposition~\ref{th:procconv} remain valid if the family of functions is extended to $\{\phi_1, \phi_{2,x}, \phi_{3,y}, \psi : x\ge x_0, y\ge y_0\}$. Hence the assertion follows from
  \begin{eqnarray*}
    \hat\alpha_n -\alpha_n & = & \frac{nv_n+(nv_n)^{1/2} \tilde Z_n(\phi_1)}{n \texpec{\tilde\psi(X_0/u_n) \, \1(X_0>u_n)}+(nv_n)^{1/2} \tilde Z_n(\psi)+R_n}-\alpha_n\\
    & = & \alpha_n \frac{1+(nv_n)^{-1/2} \tilde Z_n(\phi_1)}{1+ \alpha_n ((nv_n)^{-1/2} \tilde Z_n(\psi)+(nv_n)^{-1}R_n)}-\alpha_n\\
    & = & \alpha_n  (nv_n)^{-1/2} \frac{\tilde Z_n(\phi_1)-\alpha_n \, \{ Z_n(\psi)+(nv_n)^{-1/2}R_n\}}{1+ \alpha_n \, \{(nv_n)^{-1/2} \tilde Z_n(\psi)+(nv_n)^{-1}R_n\}}\\
    & = & (nv_n)^{-1/2} \bigl(\alpha \tilde Z_n(\phi_1)-\alpha^2 Z_n(\psi)+o_P(1)\bigr) \bigl(1+o_P(1)\bigr), \qquad n \to \infty.
  \end{eqnarray*}
  In the last step we have used stochastic boundedness of $\tilde Z_n(\psi)$ and $\tilde Z_n(\phi_1)$, which follows from their weak convergence, and the assumption that $R_n=o_P((nv_n)^{1/2})$ as $n \to \infty$.
\end{proof}

\begin{proof}[Proof of Theorem~\ref{cor:backwardestconv}]
  A Taylor expansion of the function $t\mapsto z^t$ yields
  \[
    z^{\hat\alpha_n}-z^\alpha=z^\alpha\log(z) \, (\hat\alpha_n-\alpha)+\frac 12 z^{\alpha+\lambda(\hat\alpha_n-\alpha)} (\log z)^2(\hat\alpha_n-\alpha)^2
  \]
  for some (random) $\lambda=\lambda_{z,\alpha}\in(0,1)$. Because $z^{\tilde\alpha}(\log z)^2$ is bounded for all $\tilde\alpha$ in a neighborhood of $\alpha$ and all $z\le 1/y_0$, it follows that on the event $\{X_i/X_{i-1}\ge y_0\}$, we have
  \[ \bigg|\Big(\frac{X_{i-1}}{X_i}\Big)^{\hat\alpha_n}
- \Big(\frac{X_{i-1}}{X_i}\Big)^\alpha - (nv_n)^{-1/2}\Big(\frac{X_{i-1}}{X_i}\Big)^\alpha\log\Big(\frac{X_{i-1}}{X_i}\Big) \big(\alpha\tilde Z_n(\phi_1)-\alpha^2 \tilde Z_n(\psi)\big)\bigg|\le C(\hat\alpha_n-\alpha)^2
  \]
  for some constant $C$ depending only on $\alpha$ and $y_0$ (but not on $i$ or $n$).
Hence, by Lemma~\ref{lemma:alphahatconv}, as $n \to \infty$,
  \begin{eqnarray*}
    \lefteqn{(nv_n)^{-1/2}\sum_{i=1}^n\bigg(\Big(\frac{X_{i-1}}{X_i}\Big)^{\hat\alpha_n}
- \Big(\frac{X_{i-1}}{X_i}\Big)^\alpha\bigg) \, \1\Big(\frac{X_{i-1}}{X_i}>y,X_i>u_n\Big)}\\
  & = & \frac{\alpha\tilde Z_n(\phi_1)-\alpha^2 \tilde Z_n(\psi)}{nv_n} \sum_{i=1}^n \Big(\frac{X_{i-1}}{X_i}\Big)^\alpha \log\Big(\frac{X_{i-1}}{X_i}\Big) \, \1\Big(\frac{X_i}{X_{i-1}}>y,X_i>u_n\Big)\\
  & & \hspace*{3cm} \mbox{} + o_P\Big( (nv_n)^{-1}\sum_{i=1}^n
   \1(X_i>u_n)\Big) \\
  & = &  \frac{\alpha\tilde Z_n(\phi_1)-\alpha^2 \tilde Z_n(\psi)}{nv_n}  \sum_{i=1}^n \Big(\frac{X_{i-1}}{X_i}\Big)^\alpha \log\Big(\frac{X_{i-1}}{X_i}\Big) \, \1\Big(\frac{X_i}{X_{i-1}}>y,X_i>u_n\Big) +o_P(1).
  \end{eqnarray*}
  The last equality follows from the weak convergence of $\tilde Z_n(\phi_1)$.

  As in the proof of Proposition~\ref{th:procconv}, one may establish weak convergence of
  \begin{multline*}
    \Biggl(
    (nv_n)^{-1/2}
    \sum_{i=1}^n
    \Big(\frac{X_{i-1}}{X_i}\Big)^\alpha
    \log \Big(\frac{X_{i-1}}{X_i}\Bigr) \,
    \1\Big( \frac{X_{i-1}}{X_i}>y,X_i>u_n\Big)\\
    -
    \expec{
      \Big(\frac{X_{i-1}}{X_i}\Big)^\alpha
      \log \Big( \frac{X_{i-1}}{X_i} \Big) \,
      \1\Big( \frac{X_{i-1}}{X_i}>y,X_i>u_n\Big)
    }
    \Biggr)_{y\ge y_0}
  \end{multline*}
  to a centered Gaussian process. In particular, as $n \to \infty$,
  \begin{eqnarray*}
    \lefteqn{(nv_n)^{-1} \sum_{i=1}^n \Big(\frac{X_{i-1}}{X_i}\Big)^\alpha \log\Big(\frac{X_{i-1}}{X_i}\Big) \, \1\Big( \frac{X_{i-1}}{X_i}>y,X_i>u_n\Big)}\\
    & = & \expec{ \Big(\frac{X_{-1}}{X_0}\Big)^\alpha \log  \Big(\frac{X_{-1}}{X_0}\Big) \, \1\Big( \frac{X_0}{X_{-1}}>y\Big) \, \Big| \, X_0>u_n}+o_P(1)\\
    & \dto & \expec{ \Theta_{-1}^\alpha \log( \Theta_{-1}) \, \1( 1/ \Theta_{-1}>y)}=:\ell_y,
  \end{eqnarray*}
  uniformly for $y\ge y_0$. By the time-change formula \eqref{eq:timechange} with $i=-1$, $s=t=0$ and $f(x)=-\log(x) \, \1_{(y,\infty)}(x)$, one has
  $\ell_y=-\expec{\log(\Theta_1) \, \1(\Theta_1>y)}$.

  It follows that
  \begin{multline*}
    (nv_n)^{-1/2}\sum_{i=1}^n \left( \Big(\frac{X_{i-1}}{X_i}\Big)^{\hat\alpha_n}
 \1\Big( \frac{X_i}{X_{i-1}}>y,X_i>u_n\Big)
  - \expec{\Big(\frac{X_{i-1}}{X_i}\Big)^\alpha
 \1\Big( \frac{X_i}{X_{i-1}}>y,X_i>u_n\Big)}\right)\\
   =  \tilde Z_n(\phi_{3,y})+\ell_y \, \big(\alpha Z_n(\phi_1) -\alpha^2\tilde Z_n(\tilde\psi)\big)+o_P(1), \qquad n \to \infty.
  \end{multline*}
  Proceed as in the proof of Proposition~\ref{th:procconv} to arrive at the assertion.
\end{proof}

\section*{Acknowledgments}

H.~Drees was supported by the ``Deutsche Forschungsgemeinschaft'', project JA2160/1.
J.~Segers was supported by contract ``Projet d'Act\-ions de Re\-cher\-che Concert\'ees'' No.\ 12/17-045 of the ``Communaut\'e fran\c{c}aise de Belgique'' and by IAP research network Grant P7/06 of the Belgian government (Belgian Science Policy). M.~Warcho\l{} was funded by a ``mandat d'aspirant'' of the ``Fonds de la Recherche Scientifique'' (FNRS).
 We thank two anonymous referees whose comments helped to improve the presentation of the results.

\small
\bibliographystyle{chicago}
\bibliography{mybib}

\end{document}